\documentclass{article}
\usepackage{authblk}

\usepackage{fullpage}
\usepackage[pdftex]{graphicx}
\usepackage{hyperref}
\usepackage{amssymb,amsmath,amsthm, amsfonts, centernot}
\usepackage{xcolor}
\usepackage{subcaption}

\usepackage{bbm}

\newcommand{\vkappa}{\textrm{\large$\varkappa$}}
\newcommand{\indep}{\mathrel{\text{\scalebox{1.07}{$\perp\mkern-10mu\perp$}}}}
\newcommand{\nindep}{\centernot{\indep}}

\newtheorem{theorem}{Theorem}[section]

\newtheorem{lemma}[theorem]{Lemma}

\theoremstyle{definition}
\newtheorem{definition}{Definition}[section]

\newcommand{\ti}{{\mathbf{t}}_i}
\newcommand{\tij}{{\mathbf{t}}_{ij}}

\renewcommand{\tt}{{\mathbf{t}}}

\newcommand{\ci}{{\mathbf{c}}_i}
\newcommand{\cc}{{\mathbf{c}}}
\newcommand{\cij}{{\mathbf{c}}_{ij}}
\newcommand{\cj}{{\mathbf{c}}_j}

\newcommand{\X}{{\bf\hat X}}

\newcommand{\R}{{\mathbf{R}}}

\newcommand{\BB}{{\mathcal{B}}}
\newcommand{\B}{{\mathbf{B}}}

\newcommand{\NN}{{\mathcal{N}}}

\newcommand{\HH}{{\mathcal{H}}}
\renewcommand{\H}{{\mathbf{H}}}

\newcommand{\id}{\mathbb{I}}
\newcommand{\km}{\mathbb{K}}

\newcommand{\beq}{\begin{equation}}
\newcommand{\eeq}{\end{equation}}
\newcommand{\beqa}{\begin{eqnarray}}
\newcommand{\eeqa}{\end{eqnarray}}

\newcommand{\eps}{\varepsilon}

\newcommand{\etal}{\emph{et al.~}}
\newcommand{\ie}{\emph{i.e.~}}
\title{A normal approximation for joint frequency estimatation under Local Differential Privacy}
\author{Thomas Carette}
\affil{Sony R\&D}

\begin{document}

\maketitle

%
%
\abstract{
	In the recent years, Local Differential Privacy (LDP) has been one of the corner stone of privacy preserving data analysis. However, many challenges still opposes its widespread application. One of these problems is the scalability
	of LDP to high dimensional data, in particular for estimating joint-distributions. In this paper, we develop an approximate estimator for frequency joint-distribution estimation under so-called pure LDP protocols~\cite{Wangetal17a}.
}

\section{Introduction}

With the advances of data-driven technologies in the last decades, privacy has become a hot topic at all levels of society.
Correspondingly, the litterature on methods that allow statistical analysis and inference while protecting the anonymity of the data subjects has exploded.

One such technique is Differential Privacy~\cite{Dwork14} and its extension, Local Differential Privacy (LDP)~\cite{Kasiviswanathan11}. The key contribution of (Local) Differential Privacy is to quantify, with parameters $\epsilon$ and $\delta$, the degree of anonymization of disclosed data. The parameters $\epsilon$ and $\delta$ are usually chosen 
depending on the problem at hand. One of its features is to ensure \emph{plausible deniability} which is at 
the root of the invention of the \emph{Randomized Response}~\cite{Warner65} method in 1965. Randomized Response (RR), applied on binary data, and its generalization to categorical data (GRR), are prototypical LDP protocols.

Since the introduction of the LDP formalism by Dwork and co-workers, a vast literature has derived from marrying its methodology with (G)RR, as attested by two recent comprehensive reviews~\cite{Xiongetal20, Yangetal20}. While the privacy preserving collection of a single discrete information is relatively well understood, generic solution to the problem of accessing multi-dimensional statistics can still be improved. 

One of the corner stone of recent research stems from the observation that splitting the privacy budget $\epsilon$ for concurrent question leads to an exponentially growing variance, but splitting the population so that each user answers a single question only leads to a linear growth of the variance~\cite{Wangetal17a}. Of course, the downside of the latter approach is that the correlation between the different answers is lost. Various algorithms have been proposed to collect joint-distribution using LDP, based on expectation maximization~\cite{Fantietal16,Xuebin16} or other approaches~\cite{Cormodeetal18}. Both methods suffer from scalability issues in general. On the question of $k$-way marginal release, which is core to a series of statistical methods, the CALM method~\cite{Zhangetal18} was introduced to try to strike a balance between the two strategies. However this balance remains dataset- and task-dependent, meaning higher level algorithms need to be designed.

Another question, and opportunity, lies in how to efficiently encode a complex dataset. Complexity can stem from large or unknown dictionary sizes as well as from high dimensionality and low entropy. Any categorical dataset can be turned one-dimensional at the cost of an exponentially large dictionary.

When using binary encoding of the dictionary, the randomized response can be applied on each bit independently. This approach has been explored extensively. Various method achieves good or optimal variance-utility tradeoffs for frequency estimation~\cite{Wangetal16, Wangetal17a}, but often rely on data transfers increasing exponentially with the number of features. Once again, some level of optimization can be done by randomizing lossy encoding over users, thereby reducing communication and focussing the use of the privacy budget for each user~\cite{Erlingsson14, BassilySmith15, Wangetal17a}, but by limiting the amount of information each user contribute, large populations are required to obtain useful signal/noise ratio on very small marginals.

In order to push back the limit on the number of features reported by single users, Ren~\etal~\cite{Xuebin16} proposed a multi-demensional version of RAPPOR~\cite{Erlingsson14, Fantietal16} by essentially concatenating the Bloom filters of each feature rather than considering the full row as an element of a dictionary. The advantage of this approach is that the size of bit array to communicate grows only linearly with the number of feature instead of exponentially. Also, this decomposition of the information allows easy access to lower-dimensional marginals without the need of suming over all the noisy cells of the full contingency table. The issue with this method is that reconstructing the final joint-distribution can be rather expensive. Indeed, the only method proposed so far was to obtain it numerically using Expectation Maximization (EM). Some acceleration or stabilisation of the convergence can be achieved by using clever initial guesses~\cite{Xuebin16, Wangetal17b}, but this approach is still unsatisfactory in most cases.

Another issue with using numerically estimated joint distribution is that it constitutes a dead-end for analytical research of new algorithms. As an illustration, the research performed on analytical description of the distribution of LDP frequency estimations was critical to investiate theoretic optimal solutions~\cite{Holohanetal16, Kairouzetal16,Wangetal17a}. However, to our knowledge, no such description exist that efficiently describes the distribution of \emph{marginals} in a multi-dimensional dataset.
Many LDP statistical learning algorithms, some of them iterative, rely on the knowledge of conditional probabilities in the dataset~\cite{Xiongetal20, Yangetal20} but because those conditional probabilities are not known analytically, researchers must rely on costly simulations and intuition.

In this paper, we tackle the problem of analytical approximation of the posterior joint distribution of a discrete dataset after perturbation. We focus on pure LDP protocols proposed by~\cite{Wangetal17a}. By solving this problem we hope to facilitate the development of new algorithms leveraging the analytical properties of the deduced equations.

\section{Background and notations}

In this paper, we will use boldface font to denote random variables or functions, and normal font for values.

\begin{table}[h!]
	\center
	\begin{tabular}{lll}
		$\varepsilon, \delta$ & traditional DP parameters & definition \ref{def_dp}\\
		\bf{A} & random function defining a LDP protocol & definition \ref{def_dp}\\
		$D, D'$ & representations of datasets & definition \ref{def_dp}\\
		$\mathcal{X}, \mathcal{Y}$ & input and output spaces of a LDP protocol  & definition \ref{def_dp}\\
		$p, q$ & traditional randomized response paraeters & definition \ref{def_pure_ldp}\\
		$S$ & Support of a random function \bf{A} & definition \ref{def_pure_ldp} \\
		$\mathbb{S}$ & Support operator corresponding to $S$ & equation \ref{eq-support-op}\\
		t & ground truth set of all joint frequencies of values in a dataset $D$ \\
		$\X$ & \emph{evidence matrix} of perturbed dataset & section \ref{sec-md-pureldp}\\
		\bf{c} & set of all marginal frequencies of $\X$ (PMDH) & equations \ref{eq-freq-def}-\ref{eq-coocc-def2}\\
		\bf{t} & estimator of ground truth computed with {\bf c} (TMDH)  & section \ref{sec-md-pureldp}\\
		$\mathcal{I}, \mathcal{J}, \mathcal{S}$ & Sets of columns in $\X$ & section \ref{sec-pmdh}\\
		$\alpha,\beta$ & size of column sets & section \ref{sec-pmdh}\\
		$\km^{(\alpha\beta)}(\mathcal{J}, \mathcal{S})$ & Subset incidence matrix for  $\mathcal{S}\subseteq\mathcal{J}$ & equation \ref{eq_k_def}
	\end{tabular}	
	\caption{Notation}
\end{table}

\subsection{Differential Privacy}

Differential privacy is defined in the context of a trusted curator that process the users data and releases anonymized reports. In differential privacy, the level of anonymity is defined using a statistical framework~\cite{Dwork14}.

\theoremstyle{definition}
\begin{definition}\label{def_dp}
An random function {\bf A}, mapping elements of its domain $\mathcal{X}$ to its image $\mathcal{Y}$, satisfies $(\eps, \delta)$-differential privacy with $\eps > 0$, if and only if for any two datasets $D, D' \in \mathbb{N}^{|\mathcal{X}|}$ satisfying $|| D - D' || = 1$, \ie that differ only in one element, and for all $y \in \mathcal{Y}$
$$
P({\bf A}(D) = y) \leq e^\eps P({\bf A}(D') = y) + \delta
$$
\end{definition}

This definition establishes a framework to \emph{measure} privacy leakage in the parameter $\eps$ with a certain slack $\delta$. In order to lift the requirement to rely on a trusted curator, Kasiviswanathan~\etal~\cite{Kasiviswanathan11} introduced a stricter definition, \emph{Local Differential Privacy}. In this case, rather than ensuring that the random algorithm has an output distribution that is close for neighboring databases, it is imposed that any two elements have a finite probability ratio to have the same image

\theoremstyle{definition}
\begin{definition}\label{def_ldp}
An random function {\bf A}, mapping elements of its domain $\mathcal{X}$ to its image $\mathcal{Y}$, satisfies $(\eps, \delta)$-local differential privacy with $\eps > 0$, if and only if for any two inputs $x, x' \in \mathcal{X}$ we have for all $y \in \mathcal{Y}$
$$
P({\bf A}(x) = y) \leq e^\eps P({\bf A}(x') = y) + \delta
$$
\end{definition}

We will focus on the cases where $\delta=0$, \ie strictly no accidental privacy leakage is allowed.

\subsection{Randomized response}

The randomized response protocol was introduced by Warner~\cite{Warner65} to guarantee anonymity in surveys. It can be defined formally on binary answers as
\beqa
P(\R_{p,q}(1) = 1) &=& q \label{def_rr_q}\\
P(\R_{p,q}(0) = 1) &=& p \label{def_rr_p}
\eeqa
where $\R_{p,q}$ represents the randomized response function.

Note that when $0 < p < q < 1$, the space of $\R_{p,q}$ form a monoid with the composition operation $\circ$, i.e. it forms an
algebra which has all the properties of a group except for the inversion.
\beqa
\R_{p (q'-p') + p', q (q'-p') + p'} &=& \R_{p, q} \circ \R_{p', q'} \qquad \textrm{(complete)}\\
\R_{p, q} &=& \R_{p, q} \circ \R_{0, 1} \qquad \textrm{(identity)}\\
\R_{p, q} \circ (\R_{p', q'} \circ \R_{p'', q''}) &=& (\R_{p, q} \circ \R_{p', q'}) \circ \R_{p'', q''}\qquad \textrm{(associative)}
\eeqa
The identity is a direct consequence of the relation of completeness and the composition operation is always associative. Note also that we need both $p = 1-q$ and $p' = 1-q'$ to have commutativity
between the composition of two randomized response functions.

\subsection{Pure LDP protocols}

Wang~\etal~\cite{Wangetal17a} introduced a family of randomized-response-like protocols, dubbed \emph{pure LDP protocols}.
\theoremstyle{definition}
\begin{definition}\label{def_pure_ldp}
An encoding or hashing function $A_{e} : \mathcal{X} \rightarrow \mathcal{Y}_e \subseteq \mathcal{Y}$, a random function $\mathbf{A}_{p} : \mathcal{Y}_e \rightarrow \mathcal{Y}$,
 and a support function $S : \mathcal{Y} \rightarrow \bar{\mathcal{X}}$ with $\bar{\mathcal{X}}$ the set of sets over $\mathcal{X}$, form a \emph{Pure LDP protocol} if we have two probability values $p<q$ such that for all $x,x'\in\mathcal{X}$, $x\neq x'$
\beqa
P(\mathbf{A}(x)  \in \{y | x \in S(y) \}) &=& q \\
P(\mathbf{A}(x') \in \{y | x \in S(y) \}) &=& p 
\eeqa
where $\mathbf{A} = \mathbf{A}_p \circ A_e$
\end{definition}
$\mathcal{X}$ is the dataset dictionary of values, while $\mathcal{Y}_e$ is the dictionary of encoded values (which can be smaller than the original dictionary). The space $\mathcal{Y}$ can be outside of the $\mathcal{Y}_e$ dictionary space so that the support function cannot be considered as a decoder.

The above definition states that components of an encoding of an input value are conserved with a probability $q$, and that there is a probability $p$ than any other input value would share the same component\footnote{Unfortunately, p and q are used interchangeably in the litterature. We adopt the convention of the RAPPOR paper~\cite{Erlingsson14} while Wang~\etal~\cite{Wangetal17a} use $p$ for the probability of conserving the original value, and $q$ for the probability of polution.}. It can be seen as a generalization of the Randomized Response protocol, which is itself a pure LDP protocol. As noted by the original authors, it is a direct consequence of the definition~\ref{def_pure_ldp} that the use of Bloom filters as encoding function, as in RAPPOR, does not lead to a pure protocol as soon as any hashing collision $h_i(x) = h'_i(x')$ occurs, $h$ and $h'$ being any hashing function of the filter.

The support operator maps perturbed vectors in the Unary Encoding space
\beq\label{eq-support-op}
\mathbb{S}_{x}(y) = 1 \textrm{ if } x \in S(y) \textrm{ and } 0 \textrm{ otherwise}
\eeq
Wang~\etal~\cite{Wangetal17a} prove that
\beq\label{eq-support-est}
\mathbf{t}_x = \frac{\sum_{i} \mathbb{S}_{x}(\mathbf{A}(x_i)) - np}{q-p}
\eeq
is an unbiased estimator for the frequency of any element $i$,
where $x, \{x_i\} \in \mathcal{X}$, and $i$ spans all the elements of the dataset. Note that the random variable $\mathbb{S}_{x}(\mathbf{A}(x_i))$ follow Bernouilli distributions
\beq
\mathbb{S}_{x}(\mathbf{A}(x_i)) \sim
\left\{\begin{array}{ll}
\textrm{Bernouilli}(q) & \textrm{if } x=x_i \\
\textrm{Bernouilli}(p) & \textrm{if } x \neq x_i \\
\end{array}\right.
\eeq
Two variables $\mathbb{S}_{x}(\mathbf{A}(x_i))$ and $\mathbb{S}_{x}(\mathbf{A}(x_j))$ obviously follow independent distributions. However, for a single input they are in general not independent, \ie $\mathbb{S}_{x}(\mathbf{y}) \nindep \mathbb{S}_{x'}(\mathbf{y})$.

\subsection{Multi-dimensionality, marginals and pure LDP protocols}\label{sec-md-pureldp}

Wang~\etal~\cite{Wangetal17a} focussed on instances of pure LDP protocols applied to element frequency estimation, but it is not restricted to this case. The advantage of considering explicitely the multi-dimensionality of the dataset is to be able to avoid the exponential growth of the dictionary with the number of categories, which poses problems for communication costs, information loss during encoding, or marginal retrieval.

A first step to fight ths curse of dimensionality is to follow in the footsteps of Ren~\etal~\cite{Xuebin16} and split the dataset row domain into independent sub-domains for each feature. The encoding function $A_E$ encodes each feature independently and concatenate the results. Similarly, the perturbation function $\mathbf{A}_P$ can be a combination of randomization steps each acting in its own feature sub-domain. Finally the support function $S$ as well can be defined as mapping each feature encoding space into the corresponding dictionary. In this setting, any arbitrary combination of feature-specific pure LDP protocol will form a global pure LDP-protocol.

If we have $f$ features, then we have for the dataset value space $\mathcal{X}$ as well as for the encoded value space $\mathcal{Y}$
\beqa
\mathcal{X} &=& \mathcal{X}^{(1)} \oplus \cdots \oplus \mathcal{X}^{(f)} \\
\mathcal{Y} &=& \mathcal{Y}^{(1)} \oplus \cdots \oplus \mathcal{Y}^{(f)} \\
\eeqa
and
\beq
\mathbb{S}_{x}(\mathbf{A}_p(A_e(x_i))) = \mathbb{S}^{(1)}_{x}(\mathbf{A}^{(1)}_p(A^{(1)}_e(x^{(1)}_i))) \oplus \cdots \oplus  \mathbb{S}^{(f)}_{x}(\mathbf{A}^{(f)}_p(A^{(f)}_e(x^{(f)}_i)))
\eeq
where the superscripts indicate the feature sub-domain and the $\oplus$ is the direct sum of the resulting elements. In the following, we will assume that a basis set is built on such a decomposition of $\mathcal{X}$ and that any index $i$ indicates a specific value $x^{(k)}$ belonging to a specific feature space $\mathcal{X}^{(k)}$, thereby allowing us to unambiguously omit the feature superscript in the notation. We define the \emph{evidence matrix} random variable
\beq
\X_{ni} = \mathbb{S}_{i}(\mathbf{A}_p(A_e(x_n)))
\eeq
For a specific $i$, it forms a sequence of independent Bernouilli experiments so that the marginal sum
\beq\label{eq-freq-def}
\ci = \sum_{n=0}^N \X_{ni}
\eeq
follows the distribution of sum of Binomial random variables, one of probability $p$ and the other of probability $q$. We can also construct higher order perturbed marginal counts
\beqa
\cij &=& \sum_{n=0}^N \X_{ni}\X_{nj}\label{eq-coocc-def}\\
\cc_{ijk} &=& \sum_{n=0}^N \X_{ni}\X_{nj}\X_{nk} \label{eq-coocc-def2}\\
\cdots \nonumber
\eeqa
Those counts also follow the distribution of sum of Binomial random variables. Considering that every possible row in $\X$ occurs with a certain probability, it is straightforward that a change of basis set in the histogram space reduces the joint distribution of all {\bf c} to a sum of multinomial distributions. This change of basis set is explicited in appendix~\ref{app-rowtype} and mulinomial properties are used in the lemma of appendix~\ref{app-c-multinom}. However, it cannot be used to derive an analytical posterior for the joint frequency distribution.

In this paper, we perform an analytical derivation of multivariate normal parameters 
for the joint distribution $\{\ci, \cij, \dots\}$, later referred to as the Perturbed Multi-Dimensional Histogram (PMDH). This allows to use Bayesian Inference to deduce the Truth Multi-Dimensional Histogram (TMDH)
\beq
P(A | B) \propto P(A) P(B|A) 
\eeq
with $A = (\{\ti, \tij, \dots\} = \{t_i, t_{ij}, \dots\})$ and $B = ()\{\ci, \cij, \dots\} = \{c_i, c_{ij}, \dots\})$.
We demonstrate the use of this model by computing analytically the distribution of the LDP correlation estimator. 

Some generalization of our results could be performed using the properties of the normal distribution under linear basis transformations. Furthermore, the developments presented in the appendix of this paper can serve as a basis for dealing with more complex cases. A relatively straightforward but tedious work would be to adapt them to LH or GRR of each independent feature space.

The decomposition of {\bf c} in independent distributions can easily be done if $\mathbb{S}_{x} \circ \mathbf{A}_p \circ A_e$ is block diagonal and that the index $i, j, k, \dots$ belong to different blocks, \ie if $\X_{ni} \indep \X_{nj} \indep \X_{nk}$.
In general, $\mathbf(A(x_n))$ is composed of multiple independent random variables. For instance one for GRR, and $|\mathcal{X}|$ for UE. If $\X_{ni} \indep \X_{nj}$, then the set of random variables they depend on is disjoint, and they compound to the probabilities $p$ or $q$ according to~\ref{def_pure_ldp}.  If $\X_{ni} \nindep \X_{nj}$ it means they depend on common random variables. For instance, using Bloom filters over set values data (to avoid collisions), two values $x$ and $x'$ can share a bit in their support function, the other bits of the filter serving to disambiguate the values. In LH and GRR, all dictionary values are supported by the same underlying multinomial random variable. 

We list a few examples of correlated co-occurrence probabilities in Table~\ref{tab-evid-corr}.
In GRR, all co-occurrence counts $\cij, \cc_{ijk}, \dots$ are simply 0 if any pair of indices belong to the same block, while in LH they are merely modulated depending on the size of the hashing dictionary.

Besides standard hashing and encoding techniques, $A_e$ can also implement various linear transformations like for instance in the Hadamard Randomized Response~\cite{Acharya19a} which also lead to dependencies between evidence matrix elements. A careful investigation of all those aspects is beyond the scope of the present work. In the following, we assume $\X_{ni} \indep \X_{nj}$.

\begin{table}
	\center
	\caption{\label{tab-evid-corr} Co-occurrence probabilities. $p_i = E(\X_{ni})$. The probability is the success probability of $\prod_{i=1}^k\X_{ni}$. LH: $d$ is the dimension of the hash space}~\\
	\begin{tabular}[h!]{lc}
		Protocol & probability\\
		UE  & $\prod_i p_i$ \\
		GRR & 0 \\
		LH  & $p_1/d^k$ \\
	\end{tabular}
\end{table}

\subsection{Previous work}

Different aspects of joint-distribution have been considered, but few concern the analytical development of marginal estimators. Fanti~\etal~\cite{Fantietal16} proposed the use an EM-algorithm to reconstruct the joint-distribution of a dataset encoded row by row. 
Ren~\etal~\cite{Xuebin16} also considered using EM, optimizing the solution by opting for a column-wise encoding scheme and speeding up convergence using an initial guess obtained with Lasso.
Wang~\etal~\cite{Wangetal17b} proposed an analytical, iterative formula for the co-occurrences frequency estimation, but they didn't derive the co-variance matrix between estimators so that EM was still necessary to obtain the joint-distribution.
Cormode~\etal~\cite{Cormodeetal18} showed the superiority of Hadamard (or Fourier) transforms to collects marginals of binary datasets.
The CALM framework of Zhang~\etal~\cite{Zhangetal18} further improves numerical derivation of marginals by offering a systematic way to neglect specific correlations, and using maximum entropy to extrapolate between lower dimensional marginals. These lower-dimensional marginals are still constructed using established methods of full row value frequency estimation.


In this work, we leverage the concept of pure LDP protocols~\cite{Wangetal17a} (as discussed in section~\ref{sec-md-pureldp}), inspired by the LoPub approach to multi-dimensional datasets~\cite{Xuebin16}, and develop those following the empirical estimation methodology~\cite{Holohanetal16, Kairouzetal16}.

\section{Normal parameters for Multi-Dimensional Histograms}

According to the Central Limit Theorem (CLT), the distribution $\mathbf{C} = \{\ci, \cij, \cc_{ijk}, \dots\}$ converges to a normal distribution for large $N$. We have 
\beqa
\lim_{N \rightarrow \infty} \mathbf{C}/N &\sim& \mathcal{P}(c)\ \NN\left(\hat \mu, \hat \Sigma; t\right)\\
\label{eq-infer-pmdh}
\lim_{N \rightarrow\infty} \mathbf{T}/N &\sim& \mathcal{P}(t)\ \NN\left(\mu, \Sigma; c\right) \label{eq-infer-tmdh}
\eeqa
where $\mathcal{P}(t)$ and $\mathcal{P}(c)$ are priors. Without any specific prior knowledge, a so-called ``non-informative prior'' is chosen over the TMDH domain. This domain is constrained in general by the fact that the individual frequencies must be comprised between 0 and 1, that counts are discrete, and that co-occurrences cannot be larger than the corresponding frequencies. Other constrains are specific to the encoding function which, for instance, could guarantee that categories in a single dimension are mutually exclusive.

In general, the dimension of this representation makes it computationally intractable, but in practice it is possible to truncate it to low-order co-occurrences as in most cases high-order co-occurrences are extremely unlikely to be statistically significant. It is however not our purpose to propose scalable algorithms using the PMDH and TMDH distributions, but rather to lay the mathematical ground work to be built upon.

In section~\ref{sec-pmdh}, we derive analytical formulas for  $\hat \mu$ and $\hat \Sigma$ for arbitrary co-occurrence order. 
In section~\ref{sec-tmdh}, we derive the associated inferred parameters $\mu$ and $\Sigma$.

For defining arbitrary co-occurrences, we will be using curved letters $\mathcal{I}$, $\mathcal{J}$, $\mathcal{S}$,\ldots as sets of indices for which bits corresponding to a single row are positive, either before or after the randomization process depending on the context. The number of elements in such a set, \ie~the number of positive
bits in the row, will be denoted $|\mathcal{I}|$.

\subsection{PMDH normal parameters}\label{sec-pmdh}

\subsubsection{Expected values}

\begin{theorem}\label{theo-exp-c}
The expected value of a co-occurrence of order $|\mathcal{I}|$ of indices $\mathcal{I} = [i_1, \dots, i_k]$ can be developed as
\beq
E(\cc_{\mathcal{I}}) = \sum_{\mathcal{S} \subset \mathcal{I}} (q-p)^{|\mathcal{S}|}p^{|\mathcal{I}|-|\mathcal{S}|}  t_\mathcal{S}
\eeq
where $t_\mathcal{\emptyset}=N$, the number of elements in the dataset.
\end{theorem}
The proof can be found in Appendix~\ref{app-exp-c}. Because the multi-variate expected value $E(\mathbf{C})$ is built out of the expected value of the marginals$\{E(\ci), E(\cij), \dots\}$ , we can write the above theorem in matrix form.
Our basis set for the multi-dimensional histogram is split in the order of co-occurrence considered: $\alpha = 1$ corresponds to basic counts $\ci$, $\alpha=2$ to pair-wise co-occurrences $\cij$, $\alpha=3$, 3rd order co-occurrences $\cc_{ijk}$, etc. Except stated otherwise, we will use greek letters $\alpha$, $\beta$, $\gamma$ to denote indices of \emph{blocks} in the histrogram vector representation, \ie range of rows and columns corresponding to co-occurrence counts of a given order.
From theorem~\ref{theo-exp-c} we have
\beq\label{eq-exp-c-vec}
\hat \mu = M \cdot T + P N 
\eeq
defining
\beqa
M &=& \sum_{\alpha=1}^{D} \sum_{\beta=1}^{\alpha} p^{\alpha-\beta}(q-p)^\beta \km^{(\alpha\beta)}\label{eq-c-exp}\\
P &=& \sum_{\alpha=1}^{D} p^{\alpha} \mathbb{V^{(\alpha)}}
\eeqa
where $\mathbb{V^{(\alpha)}}$ the vector of ones in the $\alpha$ block and 0 elsewhere. The matrix $M$ is triangular. The matrices $\km^{(\alpha\beta)}$ are binary matrices defined to map all sets of indices $\mathcal{S}$ of size $\beta$ to the set of indices $\mathcal{J}$ of size $\alpha$. It is defined element-wise as
\beq\label{eq_k_def}
\km^{(\alpha\beta)}(\mathcal{J}, \mathcal{S}) = 1\ \mathrm{iif}\ \mathcal{S} \subseteq \mathcal{J}
\eeq
An example for three columns is given in appendix~\ref{app-ex-m3}. Note that $\km^{\alpha\alpha}$ is the identity matrix in the $\alpha$ block and 0 elsewhere. The result above relies on the assumption that we know the dataset size, \ie $t_\emptyset = N$ ($\alpha = 0$), $V(N)=0$.

Anticipating on subsequent development, it is also possible to include it in the vector representation, hence dropping the $PN$ term of equation~\ref{eq-exp-c-vec} and to treat the knowledge of $N$ as a prior in the inference process. Alternatively, we can account for overdetermined encoding strategies. For instance, with unary encoding, we can account for a contrain $\sum_{i\in \mathcal{I}} t_i = N$ for a pertinent subset of columns $\mathcal{I}$. For a single category, the modified $M$ takes the form $\tilde M = M + b \mathbb{J}$ with $\mathbb{J}_{ij}=1$ $\forall$ $i,j$. The inverse of that matrix is easily obtained from $M^{-1}$ using the Sherman-Morrison formula
\beq\label{eq-sm-formula}
\tilde M^{-1} = M^{-1} - \frac{b\ M^{-1} \mathbb{J} M^{-1}}{ 1 + b {\bar 1}^T M^{-1} {\bar 1}}
\eeq
where ${\bar 1}_i = 1$ $\forall i$. It is trivial to extend this to any linear constrain on the $t_i$'s.

\subsubsection{Covariance matrix}

\begin{theorem}\label{theo-var-c}
The variance of a co-occurrence of order $|\mathcal{I}|$, of indices $\mathcal{I} = [i_1, \dots, i_k]$ can be developed as
\beqa
V(\cc_{I}) &=& \sum_{\mathcal{S} \subseteq \mathcal{I}} t_\mathcal{S}\ (q-p)^{|\mathcal{S}|}p^{|\mathcal{I}|-|\mathcal{S}|} (1 - (q+p)^{|\mathcal{S}|}p^{|\mathcal{I}|-|\mathcal{S}|})
\eeqa
\end{theorem}
The proof can be found in Appendix~\ref{app-var-c}.

\begin{theorem}\label{theo-covar-c}
Consider a PMDH built on a dataset of $N$ rows and of columns $\mathcal{I}$ with corresponding underlying ground truth $T = \{t_\mathcal{J} | \mathcal{J} \subseteq \mathcal{I}\}$. Consider two arbitrary column sets $\mathcal{J}$ and $\mathcal{J}'$. The covariance between the co-occurrences counts $\cc_\mathcal{J}$ and $\cc_\mathcal{J}'$ is
\beq
V(\cc_\mathcal{J}, \cc_{\mathcal{J}'})  = 
\sum_{\mathcal{S}\subseteq {\mathcal{J}_\cup}}
t_{\mathcal{S}} \ (q-p)^{|\mathcal{S}|}\ p^{|\mathcal{J}_\cup| - |\mathcal{S}|}
\left(1 - (q+p)^{|\mathcal{S} \cap \mathcal{J}_\cap|} p^{|\mathcal{J}_\cap| - |\mathcal{S} \cap \mathcal{J}_\cap|} \right)
\eeq
where $\mathcal{J}_\cup = \mathcal{J} \cup \mathcal{J}'$ and $\mathcal{J}_\cap = \mathcal{J} \cap \mathcal{J}'$.
\end{theorem}
The proof can be found in Appendix~\ref{app-covar-c}.

We can see directly that that Theorem~\ref{theo-var-c} is a special case of Theorem~\ref{theo-covar-c} with $\mathcal{J}_\cap =\mathcal{J}_\cup = \mathcal{J} = \mathcal{J}'$. Similarly, if $\mathcal{J}_\cap = \emptyset$ the correlation vanishes, as expected. Examples of low-order cases are

\beqa
V(\ci, \cij) &=& N\ p^2(1-p) + t_j\ p(q-p)(1-p)\nonumber\\ && \qquad\qquad\qquad
+ t_i\ p(q - p)[1 - q - p] + t_{ij}\ (q - p)^2(1 - q - p)\\
V(\cij, \cc_{jk}) &=& N\ p^3 (1 - p) + (t_i + t_k)\ p^2 (q-p)(1 - p) + t_j\ p^2(q-p)(1 - q - p) \nonumber\\ && \qquad\qquad\qquad\ 
(t_{ij} + t_{jk})\ p (q - p)^2(1 - q - p ) + t_{ik}\ p (q-p)^2(1 - p)\nonumber\\ && \qquad\qquad\qquad\ 
t_{ijk} (q-p)^3 (1-q-p)
\eeqa

\subsection{TMDH normal parameters}\label{sec-tmdh}

The likelihood of a certain TMDH value after observing one PMDH $C = \{c_i, c_{ij}, \dots \}$ is obtained in the normal approximation~\ref{eq-infer-tmdh} using equation~\ref{eq-exp-c-vec}
\beqa\label{eq-infer-tmdh-params}
\mu &=& M^{-1} (C - NP) \\
\Sigma &=& M^{-1} \hat \Sigma (M^{-1})^T
\eeqa
Here we must realize that $\hat \Sigma$ depends on the ground truth $T$ and hence cannot be used as-is to compute the likelihood. The simplest way is to inject the expected value $\mu$ in the computation of $\Sigma$. However, a more careful approach is sometimes necessary in the presence of prior information.

In order to compute mean and variance of the TMDH, we first deduce the form of $M^{-1}$
\begin{lemma}\label{theo-inv-m}
Given a matrix $M$ of the form~\ref{eq-c-exp}, its inverse takes the form
\beq
M^{-1} 
=
\sum_{\alpha=1}^{D} (q-p)^{-\alpha} 
\sum_{\beta=1}^{\alpha} (-p)^{\alpha - \beta}\ \km^{(\alpha\beta)}
\eeq
\end{lemma}
The proof can be found in appendix~\ref{app-inv-m}. An example for three columns is given in appendix~\ref{app-ex-m3}. 

The formulas for the mean and covariance matrix are provided in sections~\ref{sec-t-exp} and~\ref{sec-t-var} respectively and demonstrated in appendix~\ref{app-t-params}.

\subsubsection{Expected values}\label{sec-t-exp}

\begin{theorem}\label{theo-t-exp}
In estimated true co-occurrence $\mu_\mathcal{I}(C)$ on a set of columns $\mathcal{I}$
depending on an observed PMDH is
\beqa
\mu_\mathcal{I}(C) &=&
({q-p})^{-|\mathcal{I}|}
\sum_{\mathcal{S} \subseteq \mathcal{I}} (-p)^{|\mathcal{I}| - |\mathcal{S}|}c_\mathcal{S}
 \eeqa

\end{theorem}
For $|\mathcal{I}|=1$ we recover the well known frequency estimator of pure LDP protocols, and randomized response in particular.

\subsubsection{Covariance matrix}\label{sec-t-var}

The law of total variance dictates
\beqa
V(\mathbf{T})
&=& V_\mathbf{T}(E_\mathbf{C}(\mathbf{T})) + E_\mathbf{T}(V_\mathbf{C}(\mathbf{T}))\\
&=& V_\mathbf{T}(\mu) + E_\mathbf{T}(\Sigma) \\
&=& 0 + M^{-1}\hat \Sigma(E(\mathbf{T})) (M^{-1})^T \\
&=& M^{-1}\hat \Sigma(\mu) (M^{-1})^T \label{eq-tot-var}
\eeqa
where $V_\mathbf{C}$ and $V_\mathbf{T}$ are the variance at fixed $\mathbf{T}$ and $\mathbf{C}$ respectively, and $E_\mathbf{C}$ and $E_\mathbf{T}$ are similarly defined expected values. The last two steps are possible thanks to the linearity of the expected value operator. The first term vanishes because the $\mu$ estimator does not depend on ground truth. 

\begin{theorem}\label{theo-t-var}
Consider a TMDH inferred from a randomized dataset $\hat X$ of $N$ rows and of columns $\mathcal{I}$ characterized by a PMDH $C = \{c_\mathcal{J} | \mathcal{J} \subseteq \mathcal{I}\}$. Consider two arbitrary column sets $\mathcal{J}$ and $\mathcal{J}'$. The covariance between the co-occurrences counts $\tt_\mathcal{J}$ and $\tt_\mathcal{J}'$ is
\beqa
V(\tt_\mathcal{J}, \tt_{\mathcal{J}'})
 &=&
\left(\frac{1}{q-p}\right)^{|\mathcal{J}|+|\mathcal{J'}|}\!\!
\sum_{\mathcal{S} \subseteq \mathcal{J}_\cup}
c_{\mathcal{S}}
(-p)^{|\mathcal{J}_\cup| - |\mathcal{S}|}
 \left(
 (-p)^{|\mathcal{J}_\cap\setminus\mathcal{S}|}(1-2p)^{|\mathcal{J}_\cap\cap\mathcal{S}|} - (q-p)^{|\mathcal{J_\cap}|}
 \right)\nonumber\\
\eeqa
where $\mathcal{J}_\cup = \mathcal{J} \cup \mathcal{J}'$ and $\mathcal{J}_\cap = \mathcal{J} \cap \mathcal{J}'$.
\end{theorem}

We see that $\mathcal{J}_\cap = \emptyset$ lead trivially to a vanishing correlation. We can derive the following examples of variances
\beqa
V(\tt_i) &=& (q-p)^{-2}\left(qpN + (1-p-q) c_i\right)\\
V(\tt_{ij}) &=& (q-p)^{-4}\left(\left( (1-2 p)^2 - (q - p)^2\right)c_{ij} + p\left(p(1 - 2p) + (p-q)^2\right) (c_i+c_j) - p^2q(q-2p)N\right)
\nonumber\\&&\qquad
\eeqa
and co-variances
\beqa
V(\tt_{ij}, \tt_i) &=& (q-p)^{-3} \left((1-q-p)(c_{ij} - p\ c_i) + qp (c_j - p N)\right)\\
V(\tt_{ij}, \tt_{jk}) &=& (q-p)^{-4} \big((1-p-q)(c_{ijk} - p (c_{ij}+c_{jk}) + p^2 c_j) + pq(c_{ik} - p (c_i+c_k) + p^2 N) \big)
\eeqa

\section{Validation and discussion}

We provide the code for the numerical validation of the normal approximation in the repository~\cite{code}. In this paper, we summarize the key findings of our numerical validation. These were obtained using an artificial dataset of 4 independent binary variables
for which we studies the full joint distribution under a pure LDP protocol with $q=0.8$ and $p=0.1$. The dataset size was varied from $\sim 10$ to $\sim 10^5$. This arguably corresponds to higher $\epsilon$ than disirable ($\sim 2$) and smaller dataset size than typical. However we trust this choice is reasonable as it shows both the transition to well determined TMDH and the issues that one can encounter with the normal approximation for very low frequencies. The user is free to download our code to test other configurations.

\subsection{Validation of the PMDH}\label{sec-val-pmdh}

There are two main approximations in the equations~\ref{theo-exp-c} and~\ref{theo-var-c} related to omitting high order moments in the distribution : we neglect the distortion close to domain boundaries (\ie when a mean frequency is close to 0 or 1, or when a co-occurrence is close to the limits imposed by the individual frequencies) and we assume a linear error (\ie the variance of a marginal frequency is independent of the error on other marginal frequencies).  The latter can be illustrated by considering the conditional probabilities over the exact multinomial and comparing them to the normal approximation. It can be proven that for partial sum following a multinomial $\{\ci, \cj, \cij\}$, $(\cij|\ci=x, \cj=y) \sim \mathcal{H}(N, x, y)$ where $\mathcal{H}$ is a hypergeometric function. This leads to
\beq
E(\cij | \ci=x, \cj = y) =  \frac{xy}{N}
\quad ; \qquad
V(\cij | \ci=x, \cj = y) =  \frac{xy(N-x)(N-y)}{N^2(N-1)}
\eeq
while te normal normal approximation yields
\beq
E(\cij | \ci=x, \cj = y) \approx  xp_j  + yp_i  - Np_jp_i
\quad ; \qquad
V(\cij | \ci=x, \cj = y) \approx Np_ip_j(1-p_i)(1-p_j)
\eeq
where $p_i,p_j$ are the probabilities of success of $\X_i$ and $\X_j$ respectively. This issue should be kept in mind when trying to force priors as conditions on the posterior distribution directly.

We try 2 different simulation techniques : full simulation, where the binary dataset is built and perturbed explicitely, and a more efficient approach where we sample from the exact sum of multinomial distributions. As expected, both converge to the same PMDH distribution. We therefore use the latter for our tests.

Figure~\ref{fig-pmdh-normalp} shows that, as expected, the calculated normal parameters agree with the simulations within the numerical accuracy. For metrics emphasizing divergence to normal, we observe a larger error which diminishes with increasing dataset size (see Figure~\ref{fig-pmdh-non-normal}).

\begin{figure}
	\centering
	\begin{subfigure}[]{0.3\textwidth}\centering
		\includegraphics[width=\textwidth]{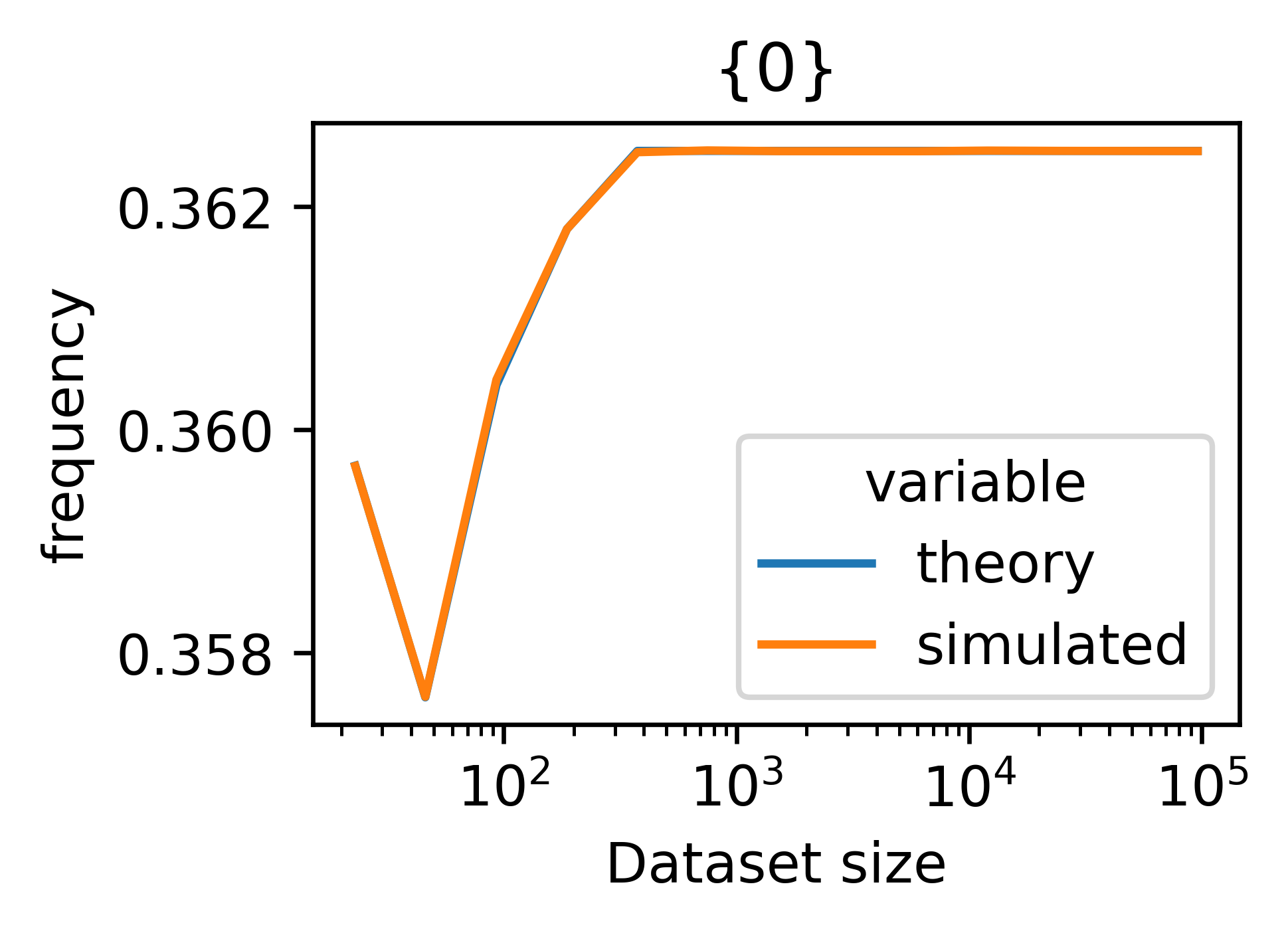}
		\caption{}
	\end{subfigure}
	\begin{subfigure}[]{0.3\textwidth}\centering
		\includegraphics[width=\textwidth]{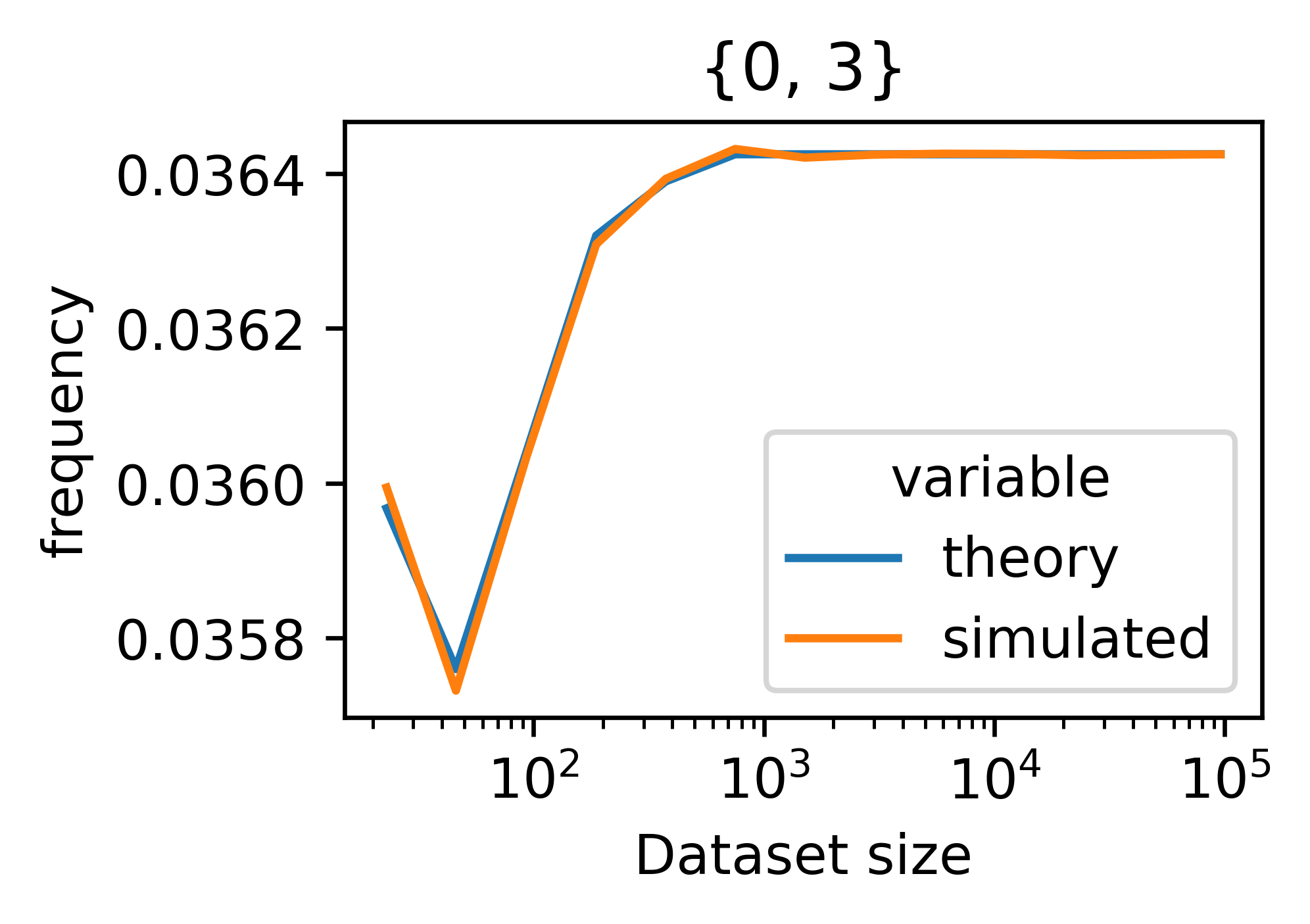}
		\caption{}
	\end{subfigure}
	\begin{subfigure}[]{0.3\textwidth}\centering
		\includegraphics[width=\textwidth]{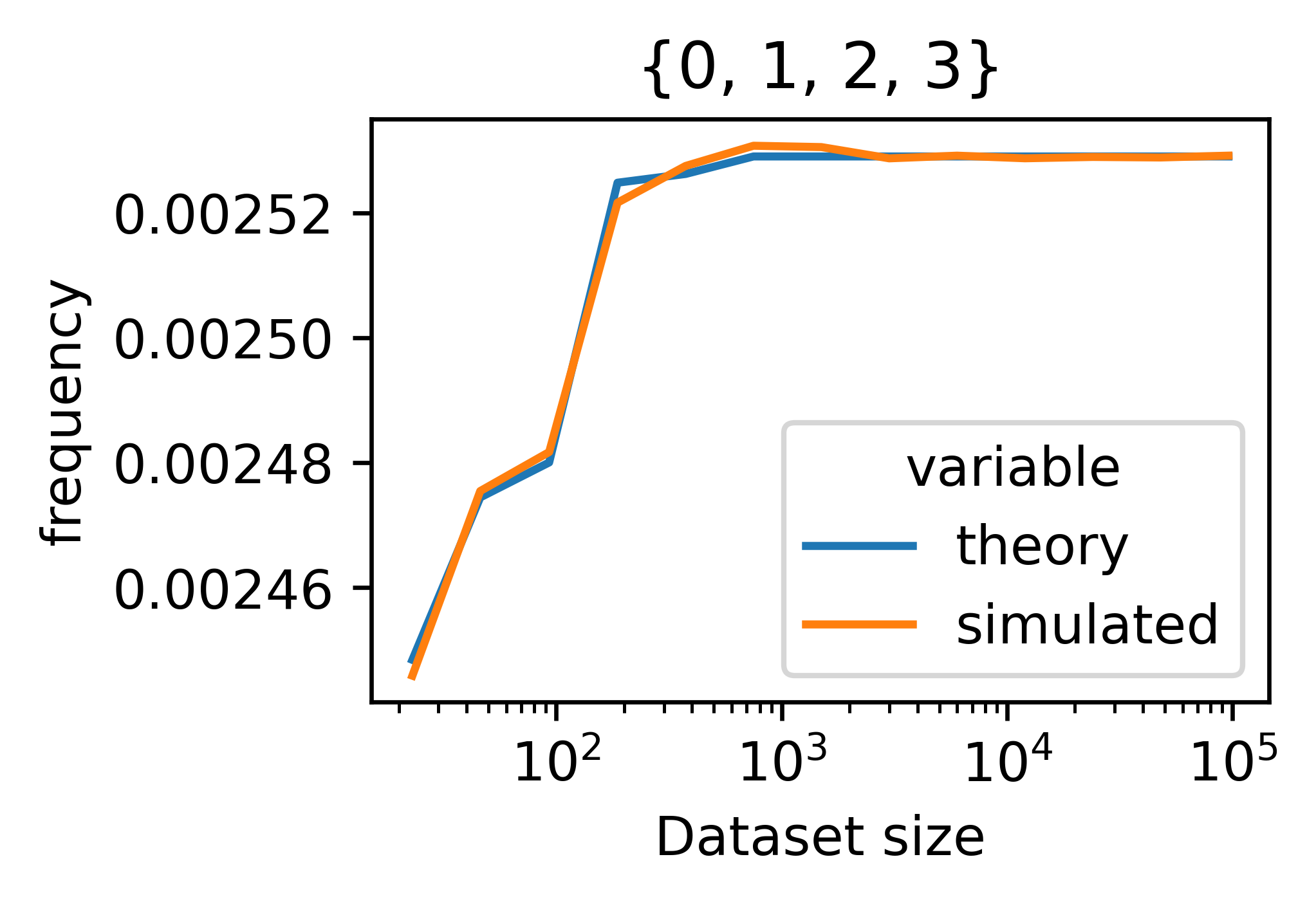}
		\caption{}
	\end{subfigure}
	
	\begin{subfigure}[]{0.3\textwidth}\centering
		\includegraphics[width=\textwidth]{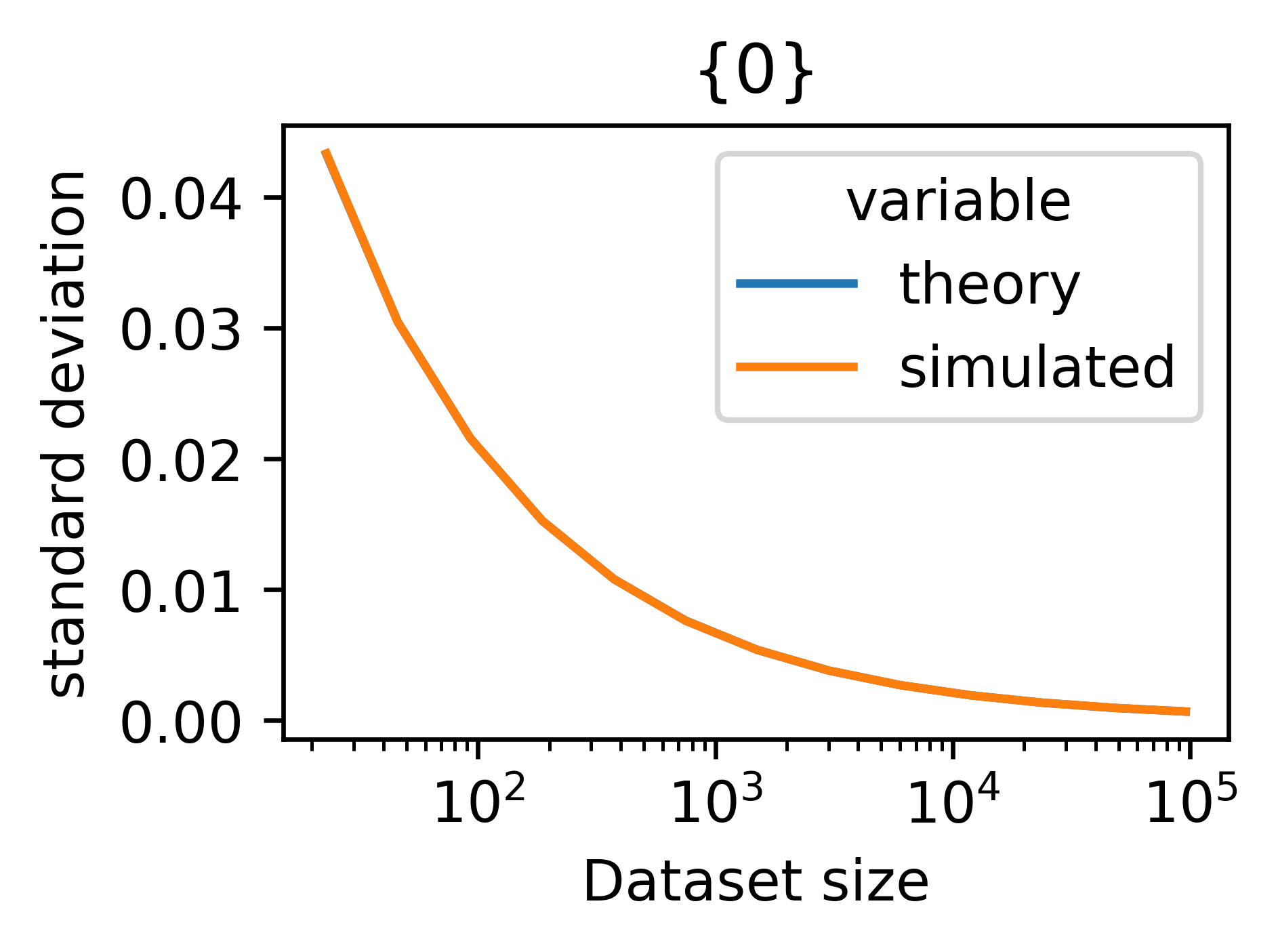}
		\caption{}
	\end{subfigure}
	\begin{subfigure}[]{0.3\textwidth}\centering
		\includegraphics[width=\textwidth]{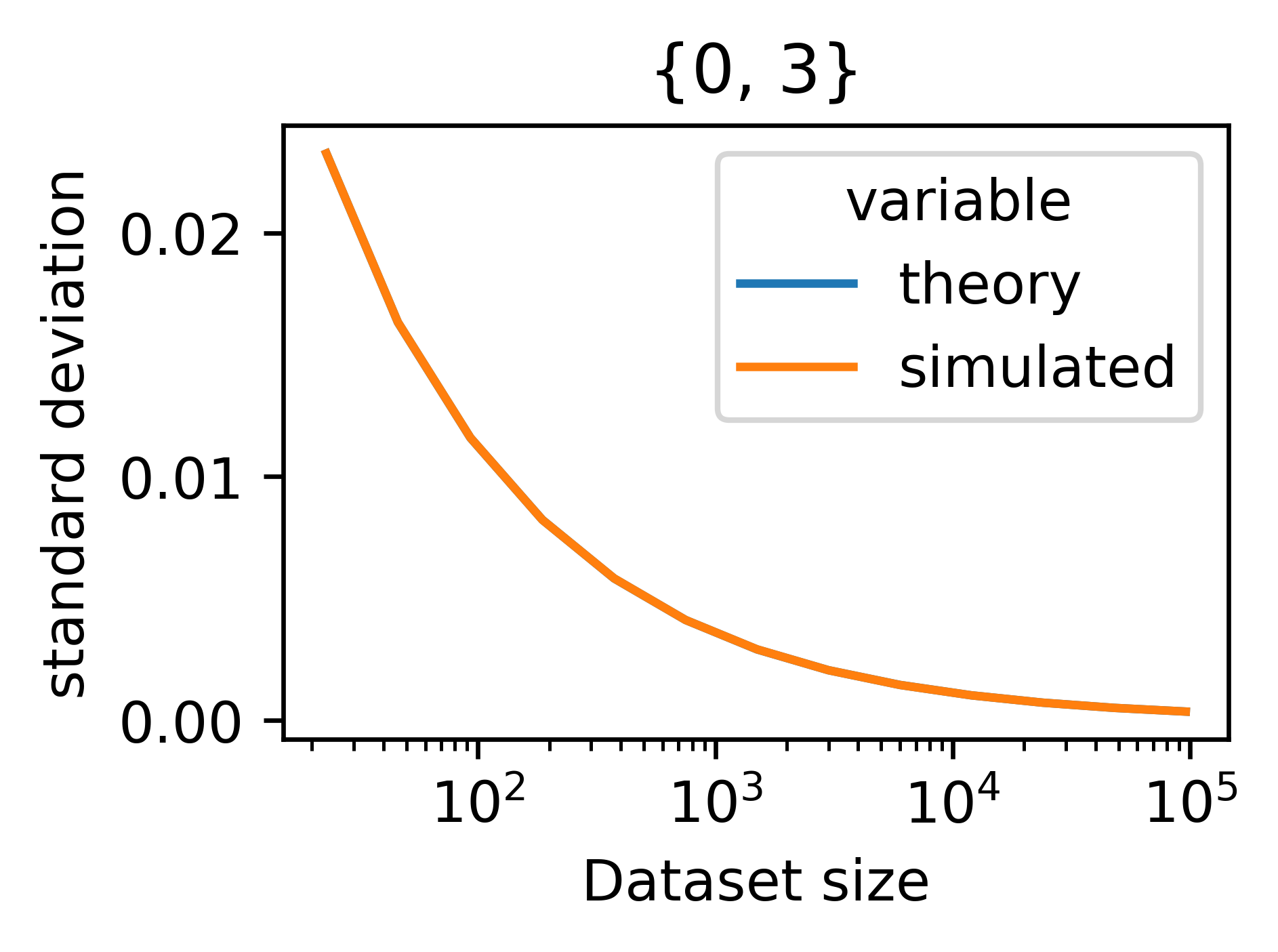}
		\caption{}
	\end{subfigure}
	\begin{subfigure}[]{0.3\textwidth}\centering
		\includegraphics[width=\textwidth]{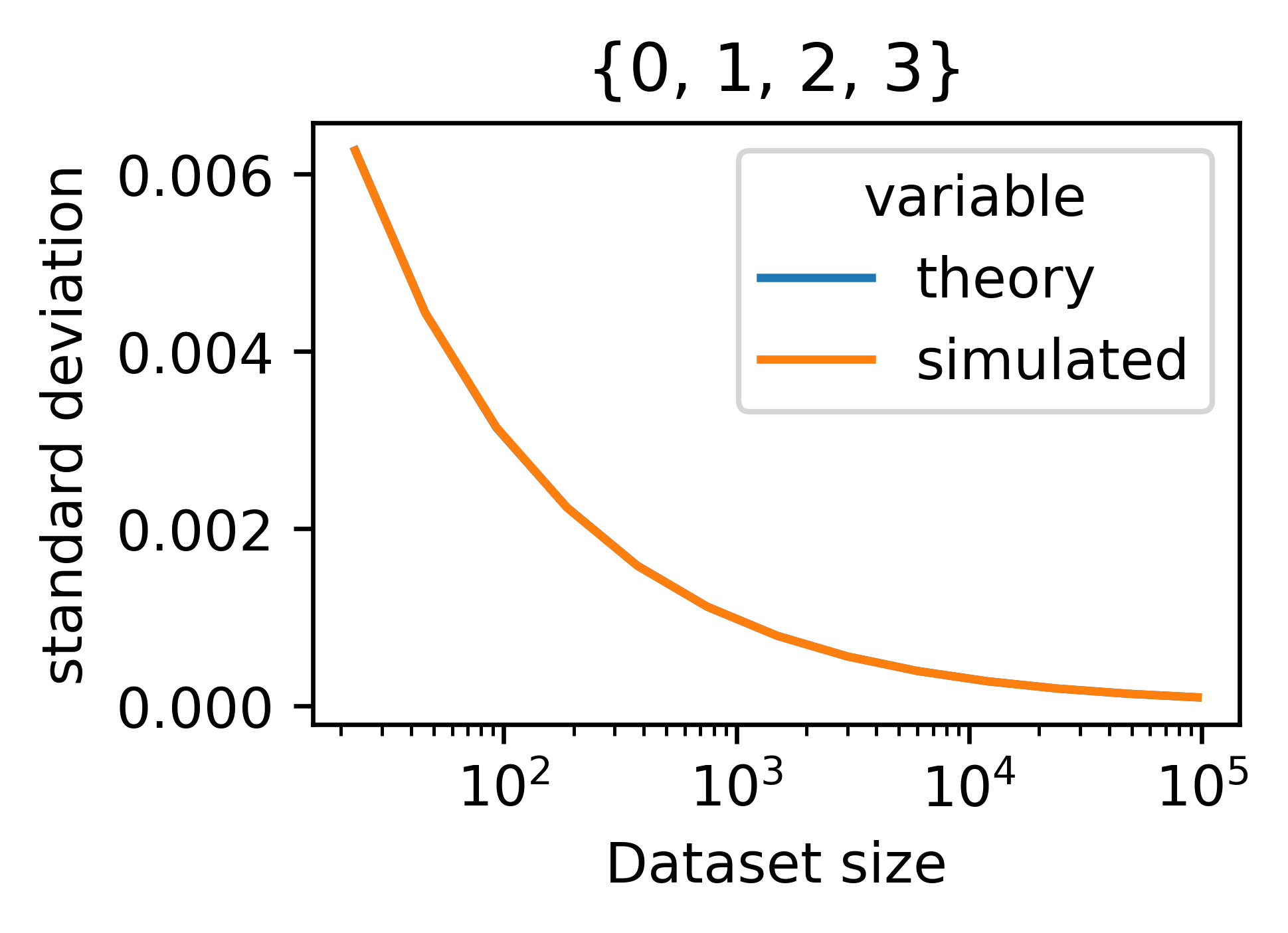}
		\caption{}
	\end{subfigure}

	\begin{subfigure}[]{0.3\textwidth}\centering
		\includegraphics[width=\textwidth]{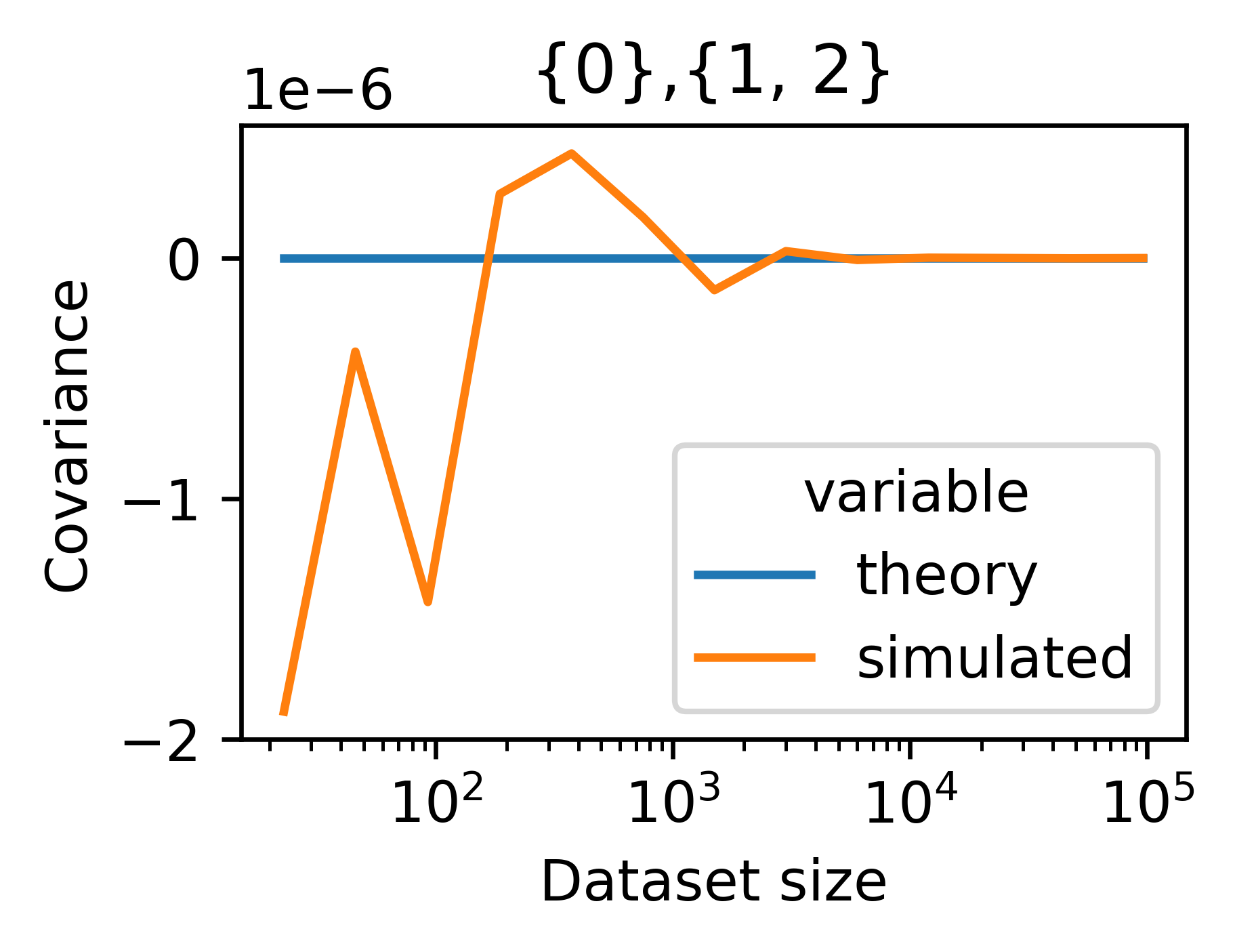}
		\caption{}
	\end{subfigure}
	\begin{subfigure}[]{0.3\textwidth}\centering
		\includegraphics[width=\textwidth]{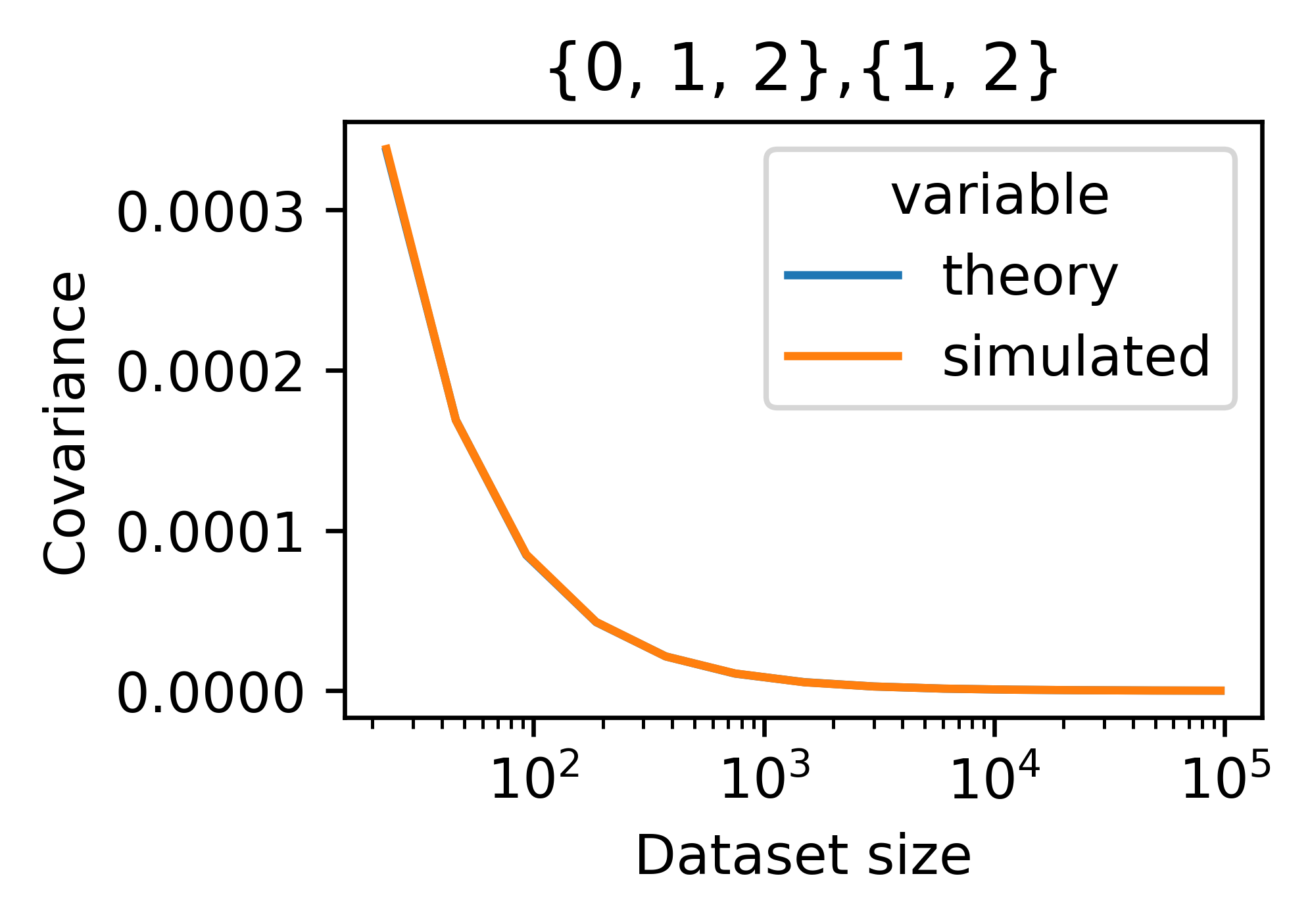}
		\caption{}
	\end{subfigure}
	\begin{subfigure}[]{0.3\textwidth}\centering
		\includegraphics[width=\textwidth]{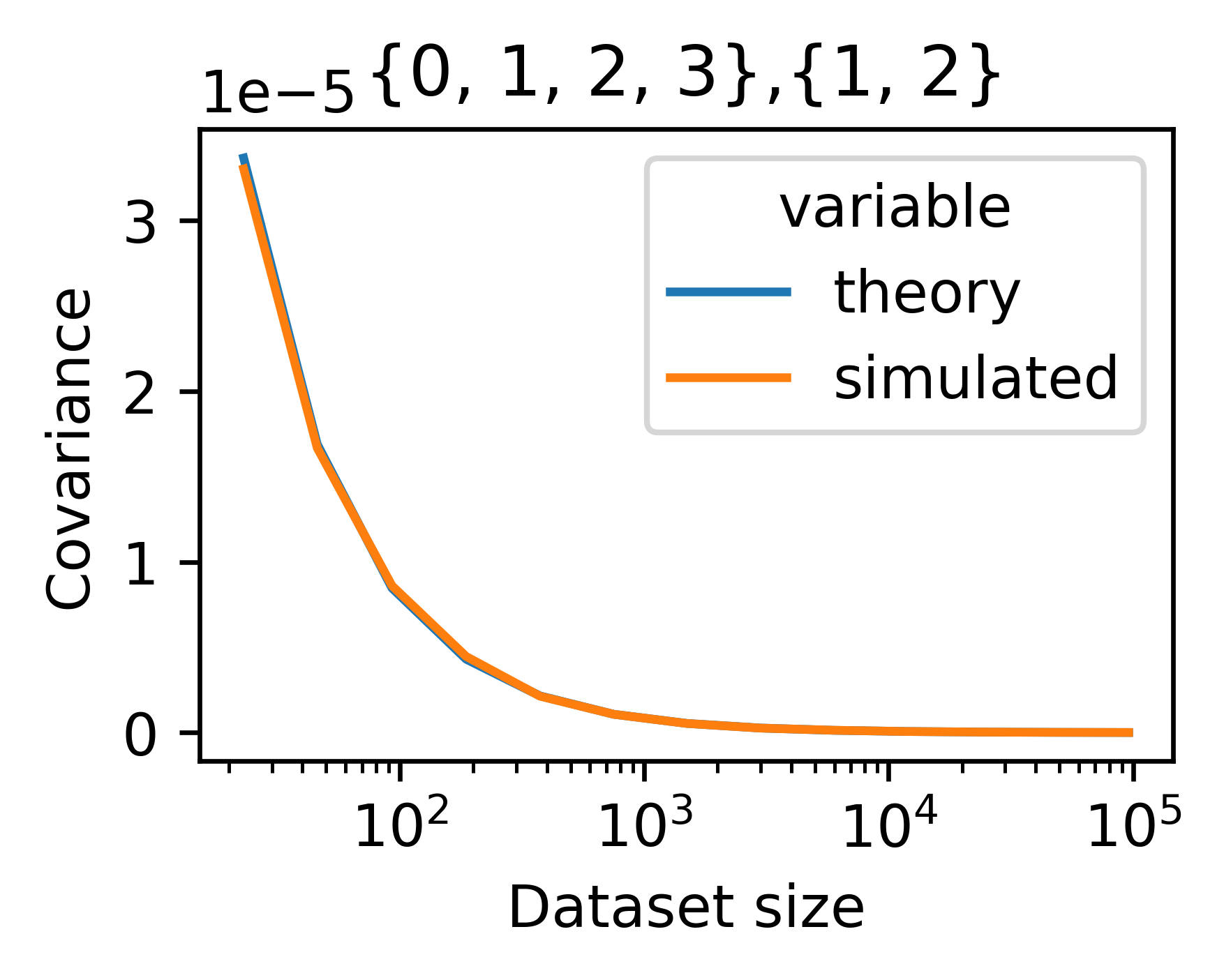}
		\caption{}
	\end{subfigure}
	\caption{Examples of mean and covariance matrix elements computed analytically vs. simulated. Bracket on top of figures signify the set of columns $\mathcal{I}$. (a-c) $E(\cc_\mathcal{I})$, (d-f) $\left(V(\cc_\mathcal{I})\right)^{1/2}$, (g-i) $V(\cc_\mathcal{I}, \cc_\mathcal{J})$ \label{fig-pmdh-normalp}}
\end{figure}

\begin{figure}
	\centering
	\begin{subfigure}[]{0.3\textwidth}\centering
		\includegraphics[width=\textwidth]{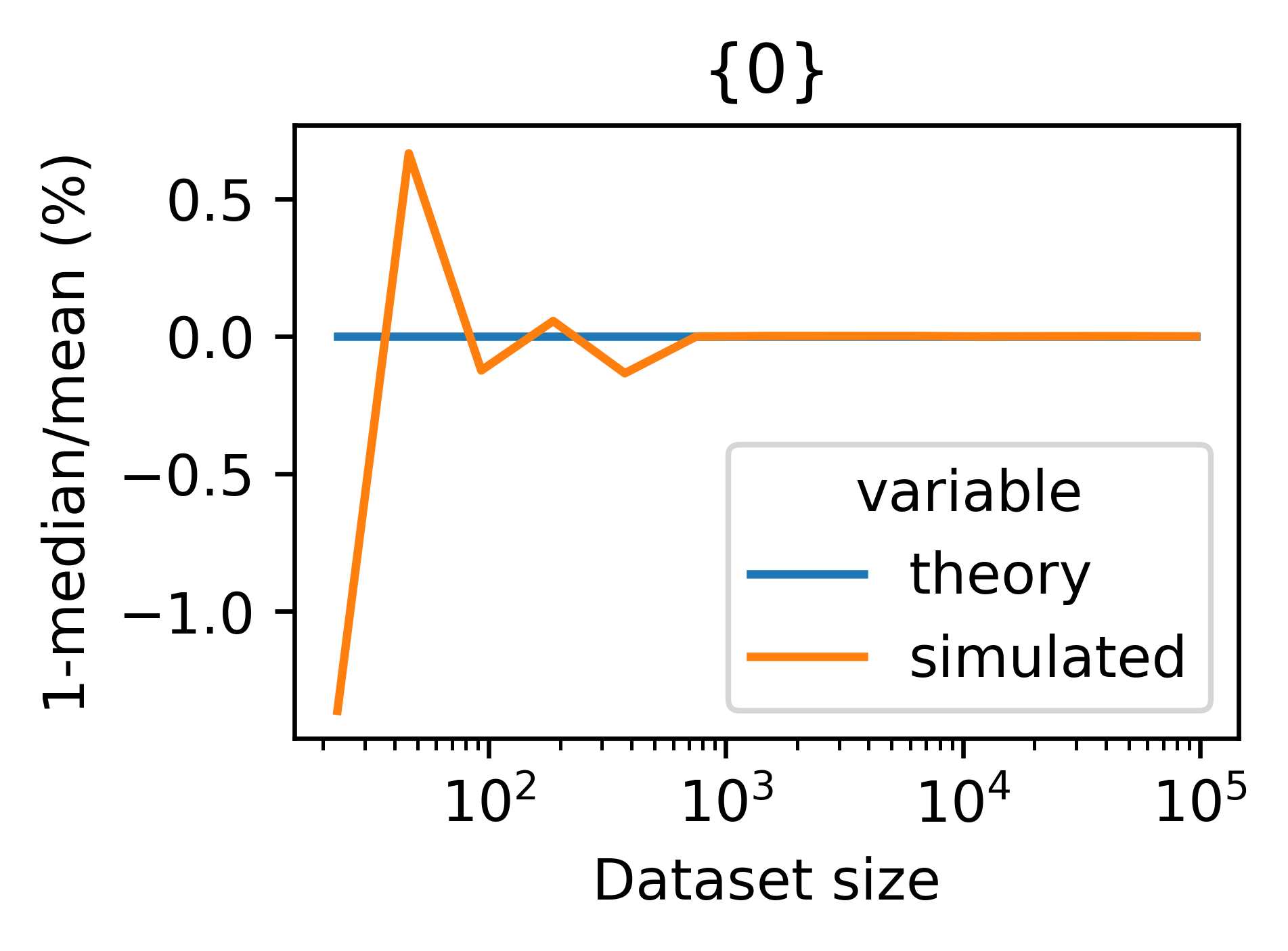}
		\caption{}
	\end{subfigure}
	\begin{subfigure}[]{0.3\textwidth}\centering
		\includegraphics[width=\textwidth]{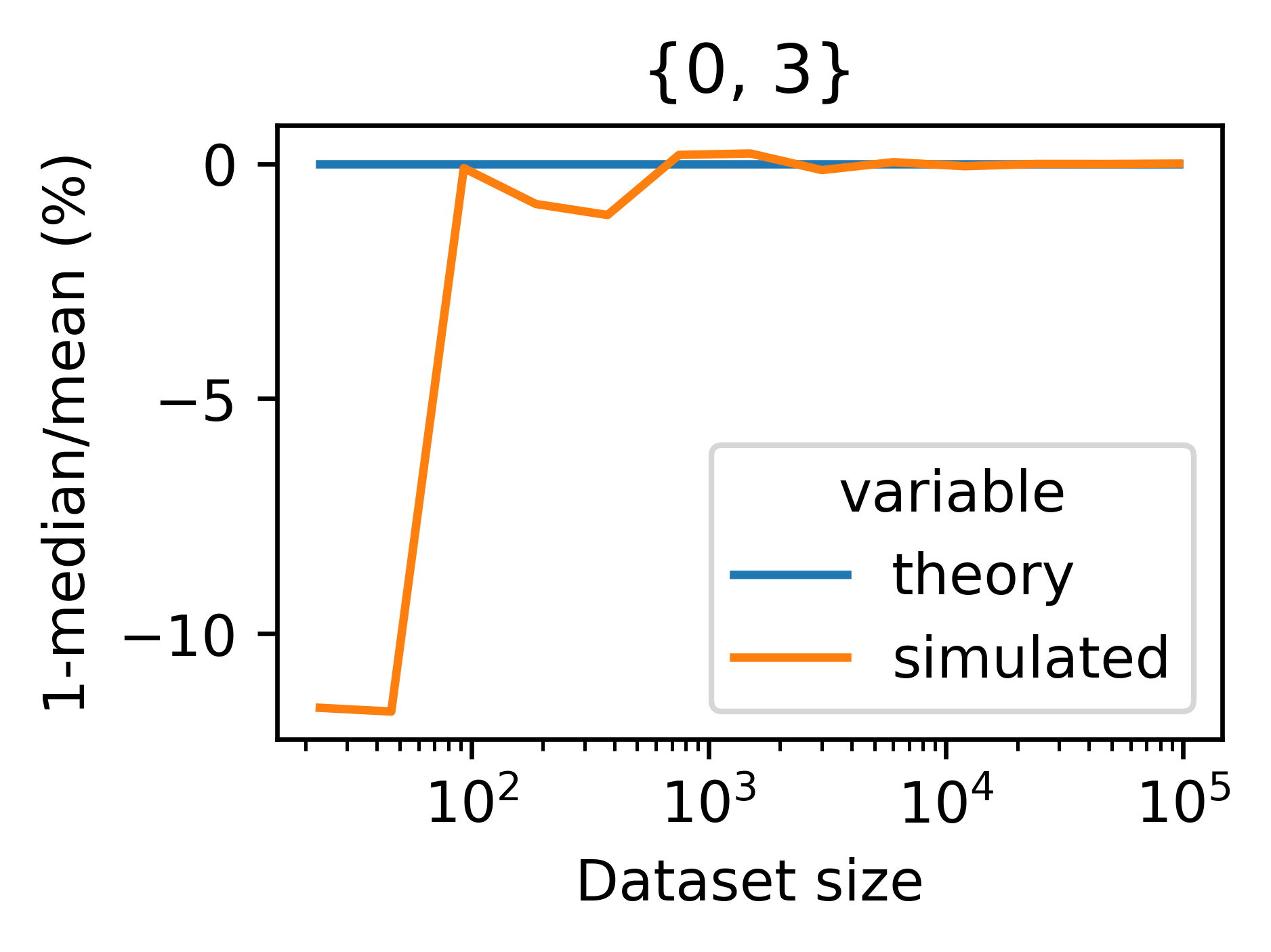}
		\caption{}
	\end{subfigure}
	\begin{subfigure}[]{0.3\textwidth}\centering
		\includegraphics[width=\textwidth]{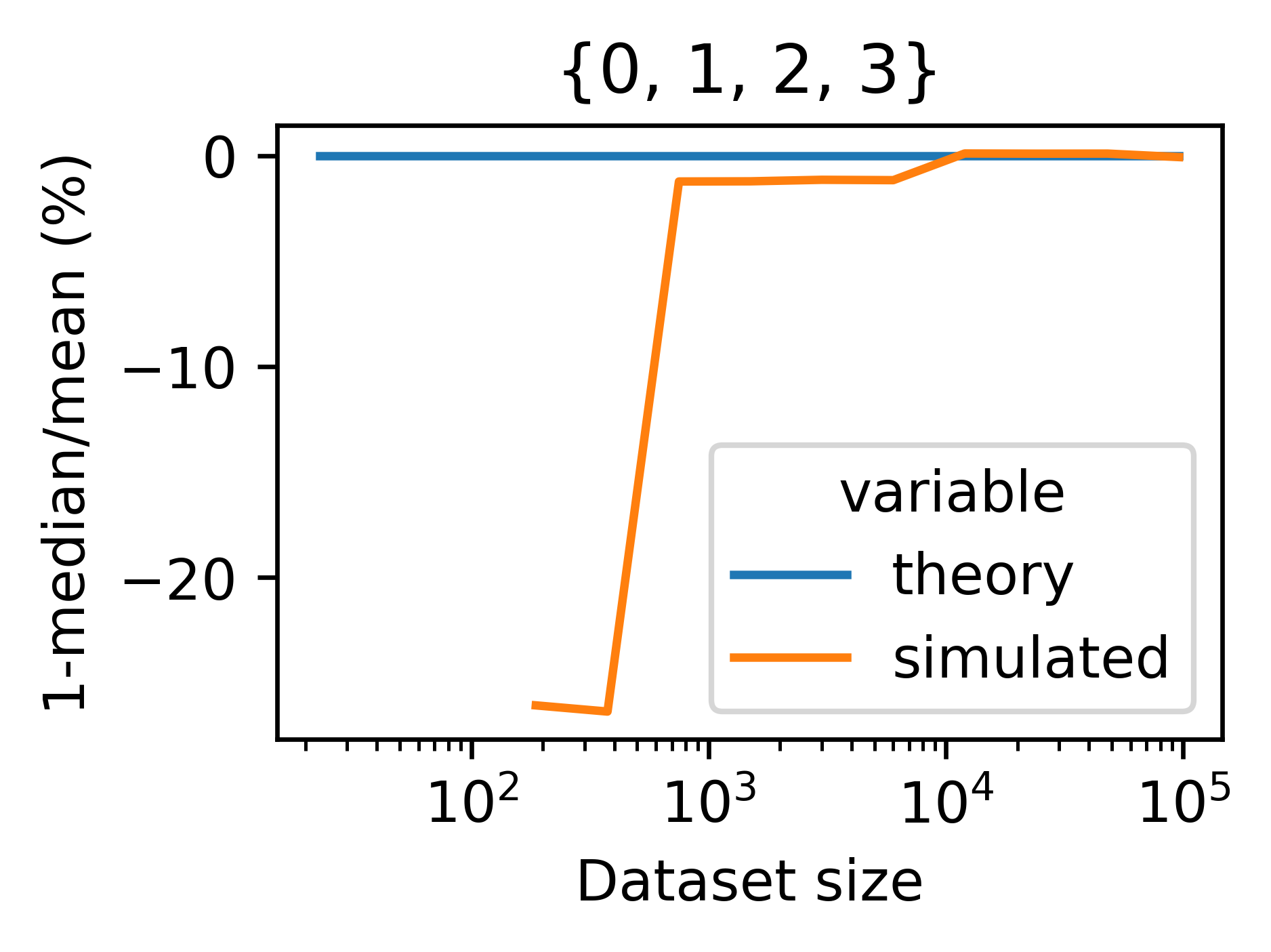}
		\caption{}
	\end{subfigure}
	
	\begin{subfigure}[]{0.3\textwidth}\centering
		\includegraphics[width=\textwidth]{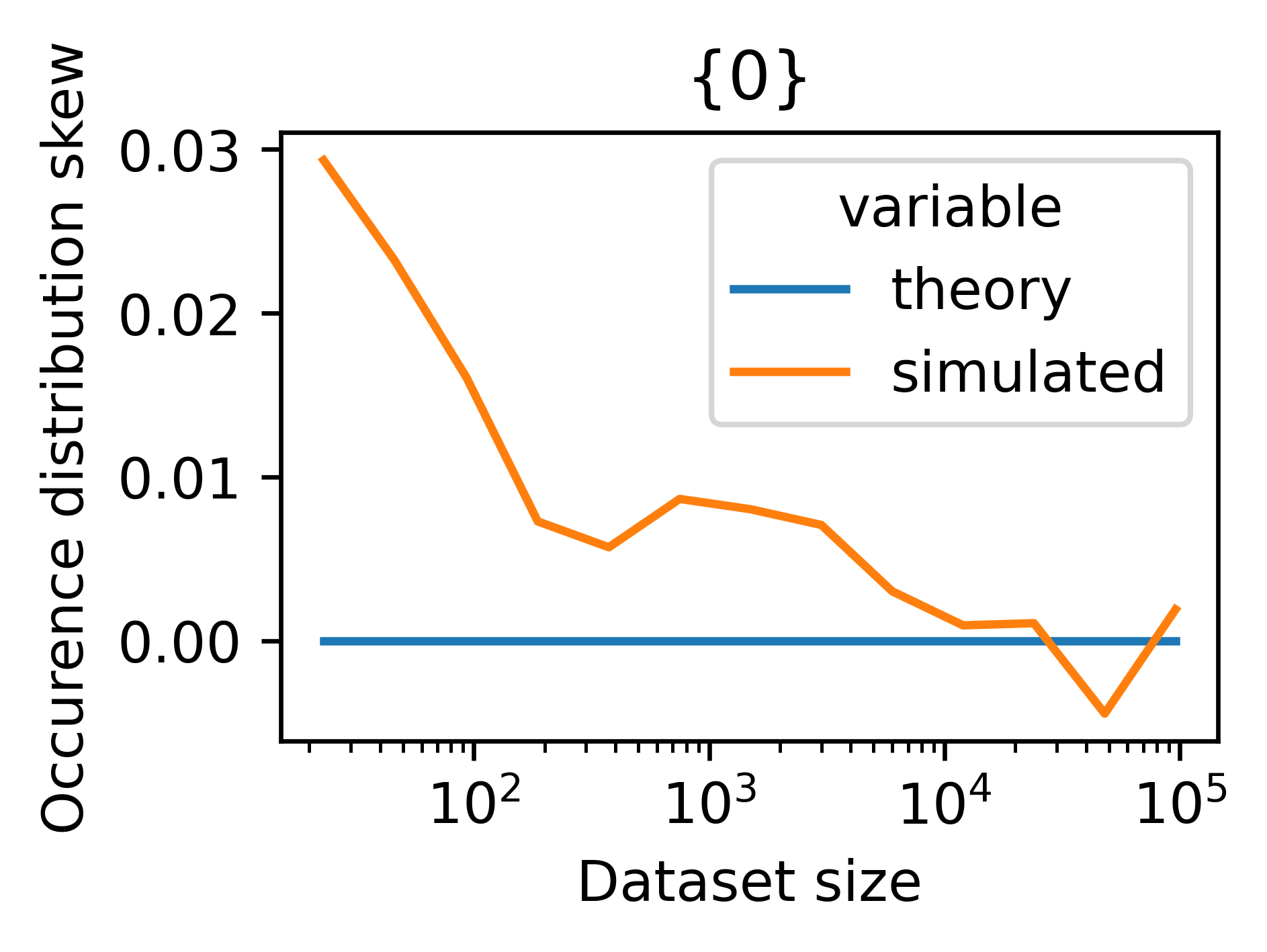}
		\caption{}
	\end{subfigure}
	\begin{subfigure}[]{0.3\textwidth}\centering
		\includegraphics[width=\textwidth]{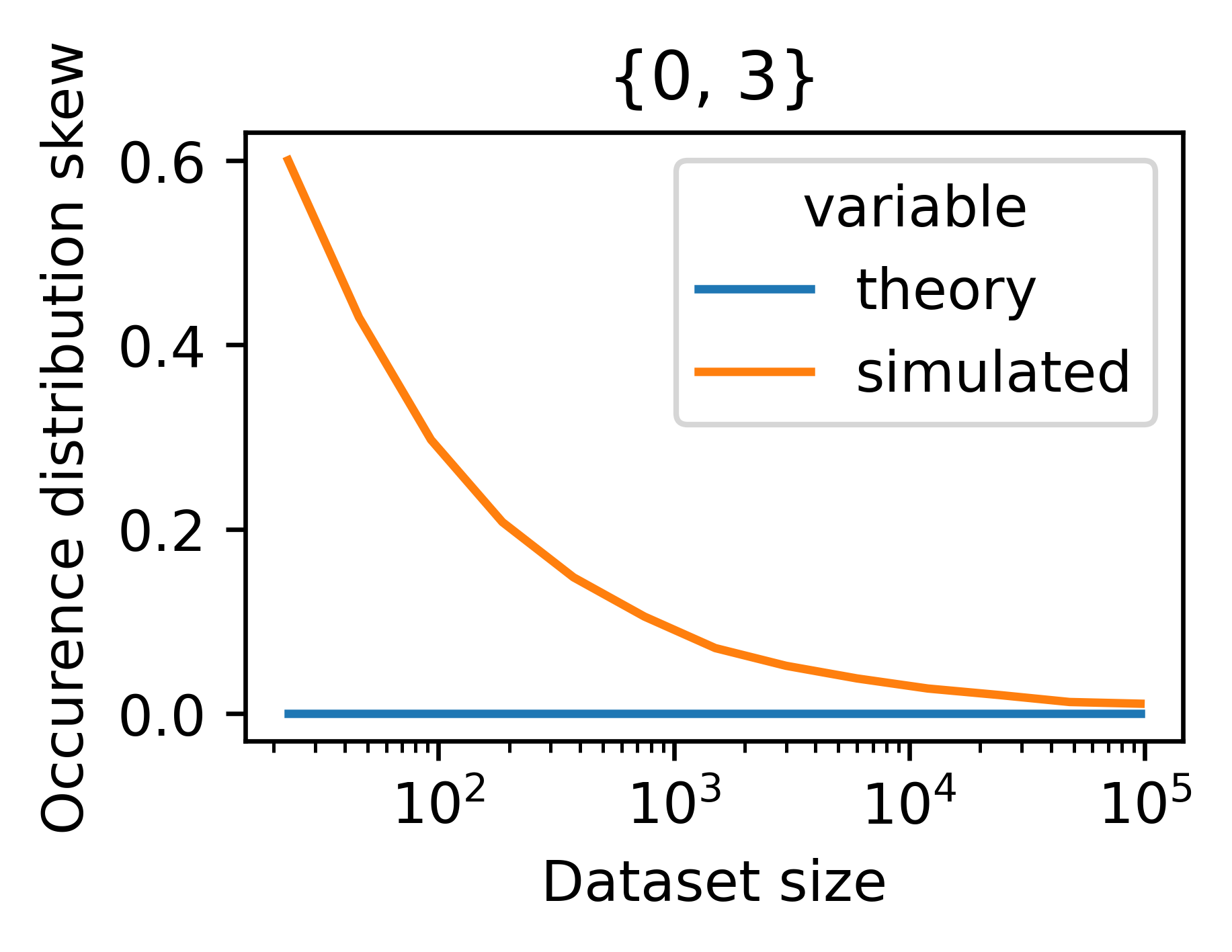}
		\caption{}
	\end{subfigure}
	\begin{subfigure}[]{0.3\textwidth}\centering
		\includegraphics[width=\textwidth]{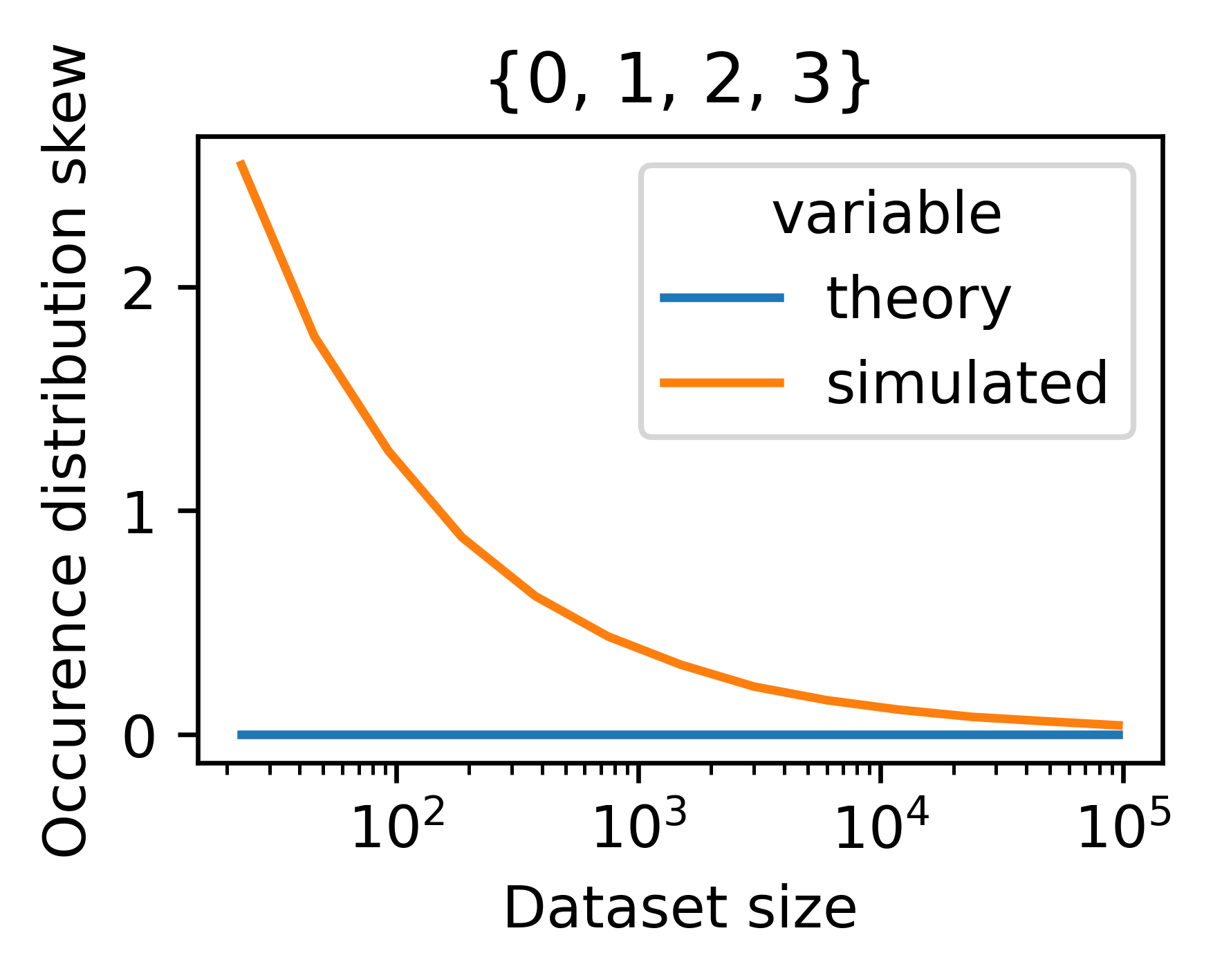}
		\caption{}
	\end{subfigure}
	\caption{Examples of metrics showing divergence to normal behavior. Bracket on top of figures signify the set of columns $\mathcal{I}$ and  $\mathcal{J}$. (a-c) relative divergence of median to mean $1-mean/median$, (d-f) $Skew(\cc_\mathcal{I})$ \label{fig-pmdh-non-normal}}
\end{figure}

\subsection{Validation of the TMDH}

To verify the quality of our estimates as a funciton of the dataset size, we sample PMDH from their exact distribution and compute the corresponding expected values for the true distribution. As in the case of the PMDH, we observe that the expected value of our mean TMDH estimator defined in~\ref{theo-t-exp} is within numerical uncertainty of the true frequencies. Furthermore, the variance of this estimator is also closely matched by the TMDH variance estimator defined in~\ref{theo-t-var}. As shown in figure~\ref{fig-tmdh-normalp-a}, we observe errors above 1\% only for low frequency $\tt_{\mathcal{I}}$ in smaller datasets ($<10^3$). In those cases, the uncertainty on the variance estimator can exceed 5\% (see figure~\ref{fig-tmdh-normalp-b}). We do not attribute this error on simulation noise, but rather on the effect of our approximation close to the boundaries of the exact distribution. This is issue is clearly visible in the figure~\ref{fig-tmdh-normalp-0-1-2-3} where the discretization of the support is visible in the reconstructed distribution. Metrics of deviation from the normal as a function of dataset size are displayed in figures~\ref{fig-tmdh-normalp}g and \ref{fig-tmdh-normalp}h.

\begin{figure}
	\centering
	\begin{subfigure}[]{0.3\textwidth}\centering
		\includegraphics[width=\textwidth]{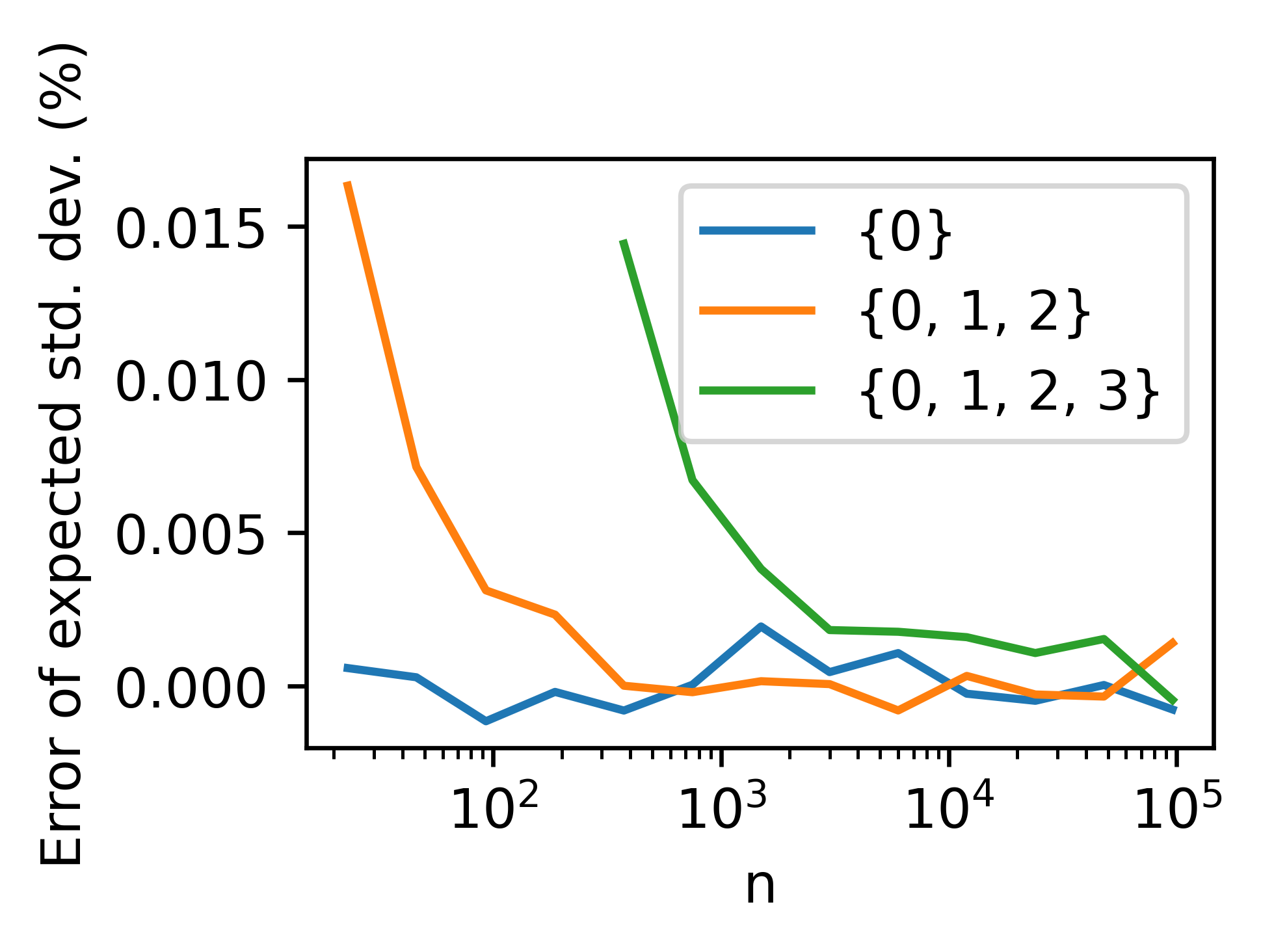}
		\caption{\label{fig-tmdh-normalp-a}}
	\end{subfigure}
	\begin{subfigure}[]{0.3\textwidth}\centering
		\includegraphics[width=\textwidth]{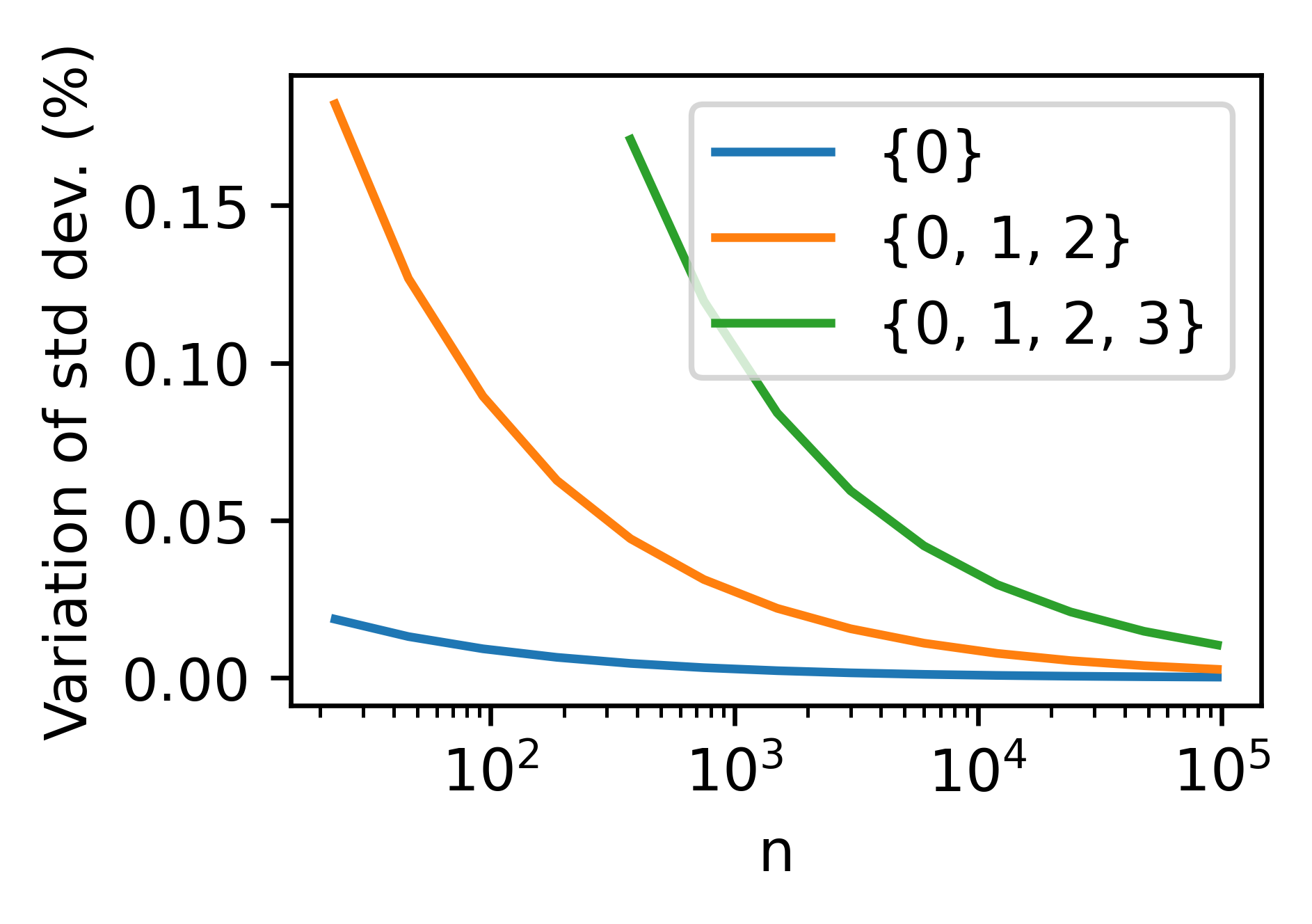}
		\caption{\label{fig-tmdh-normalp-b}}
	\end{subfigure}
	\begin{subfigure}[]{0.3\textwidth}\centering
		\includegraphics[width=\textwidth]{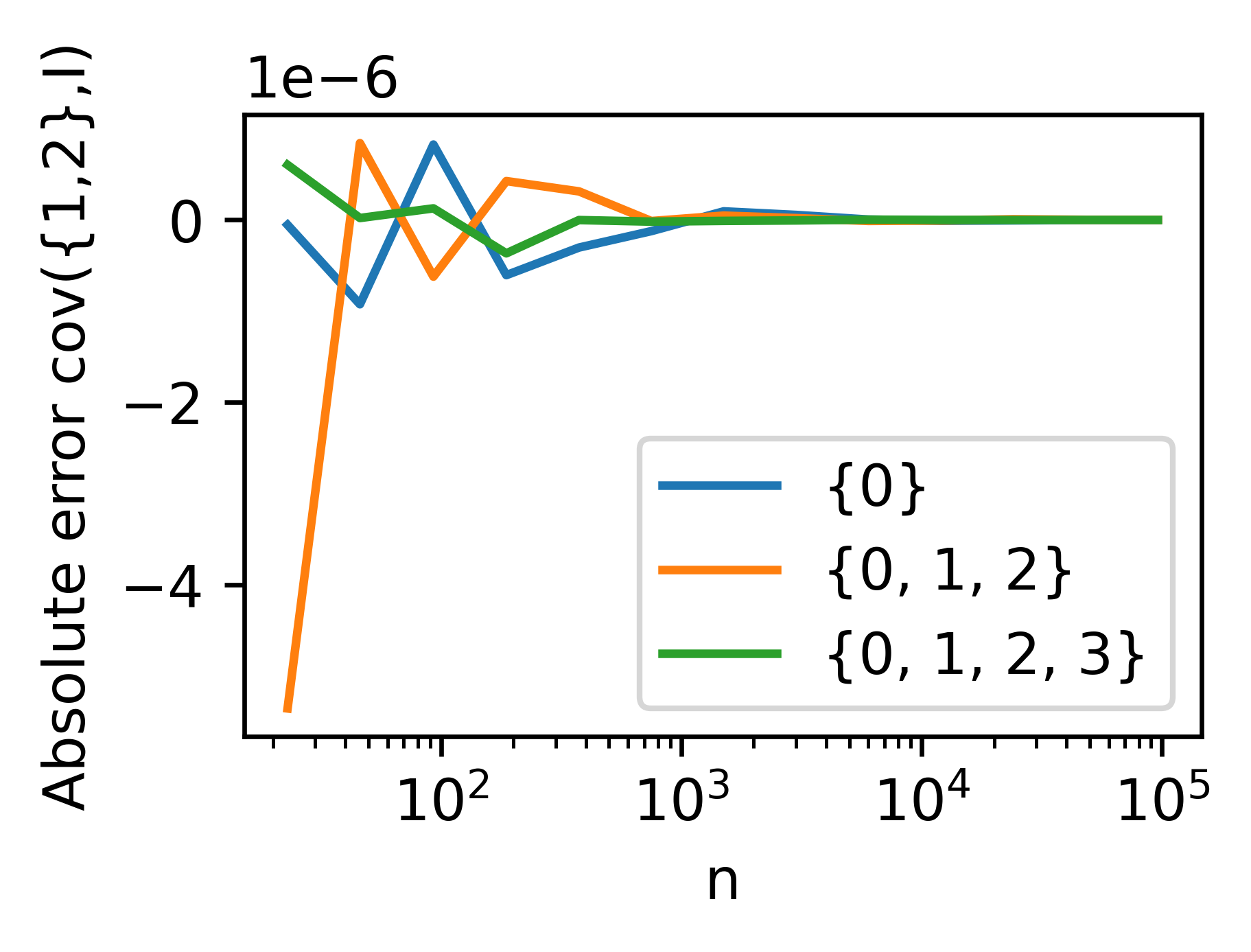}
		\caption{}
	\end{subfigure}

	\begin{subfigure}[]{0.3\textwidth}\centering
		\includegraphics[width=\textwidth]{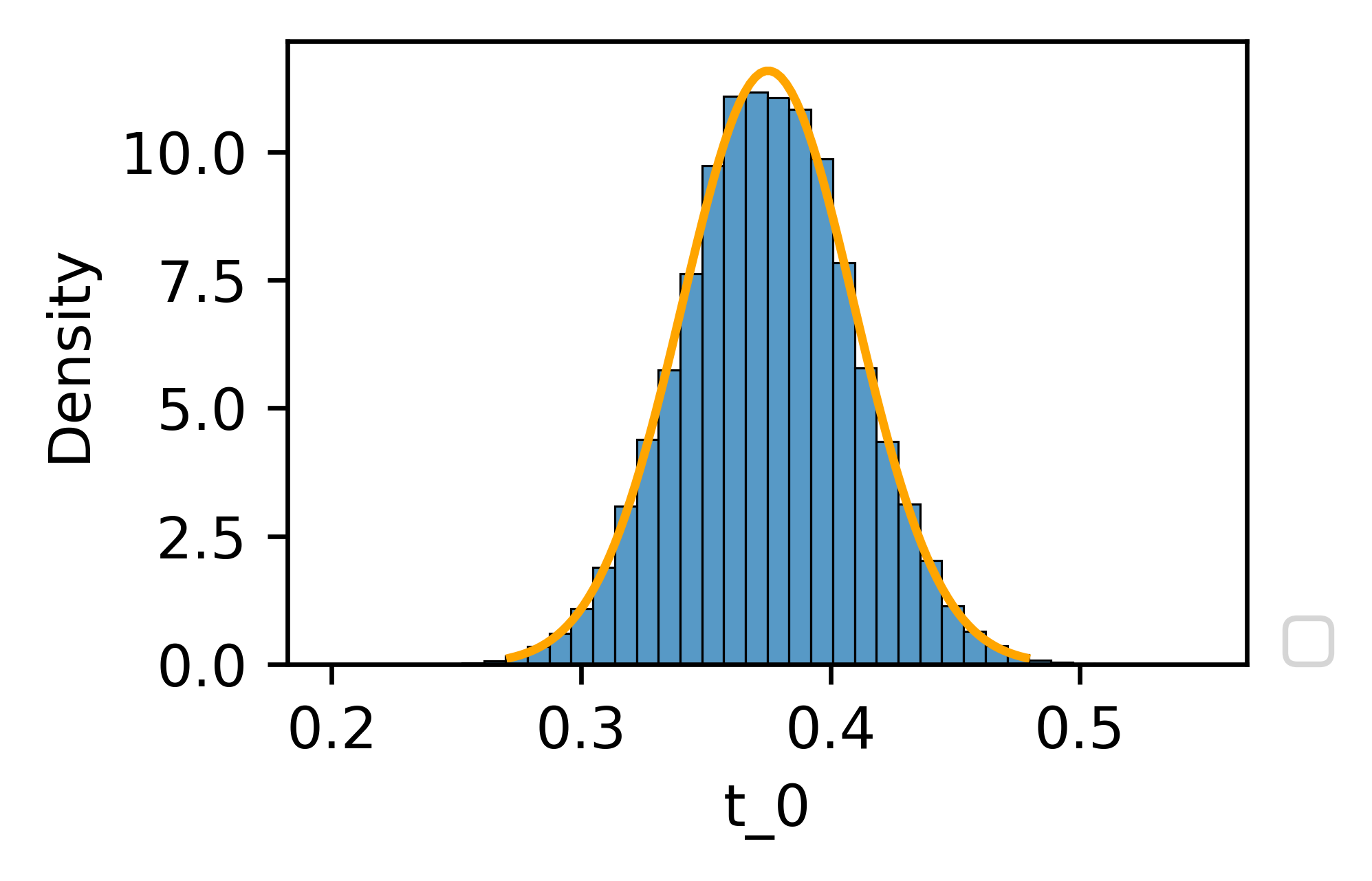}
		\caption{}
	\end{subfigure}
	\begin{subfigure}[]{0.3\textwidth}\centering
		\includegraphics[width=\textwidth]{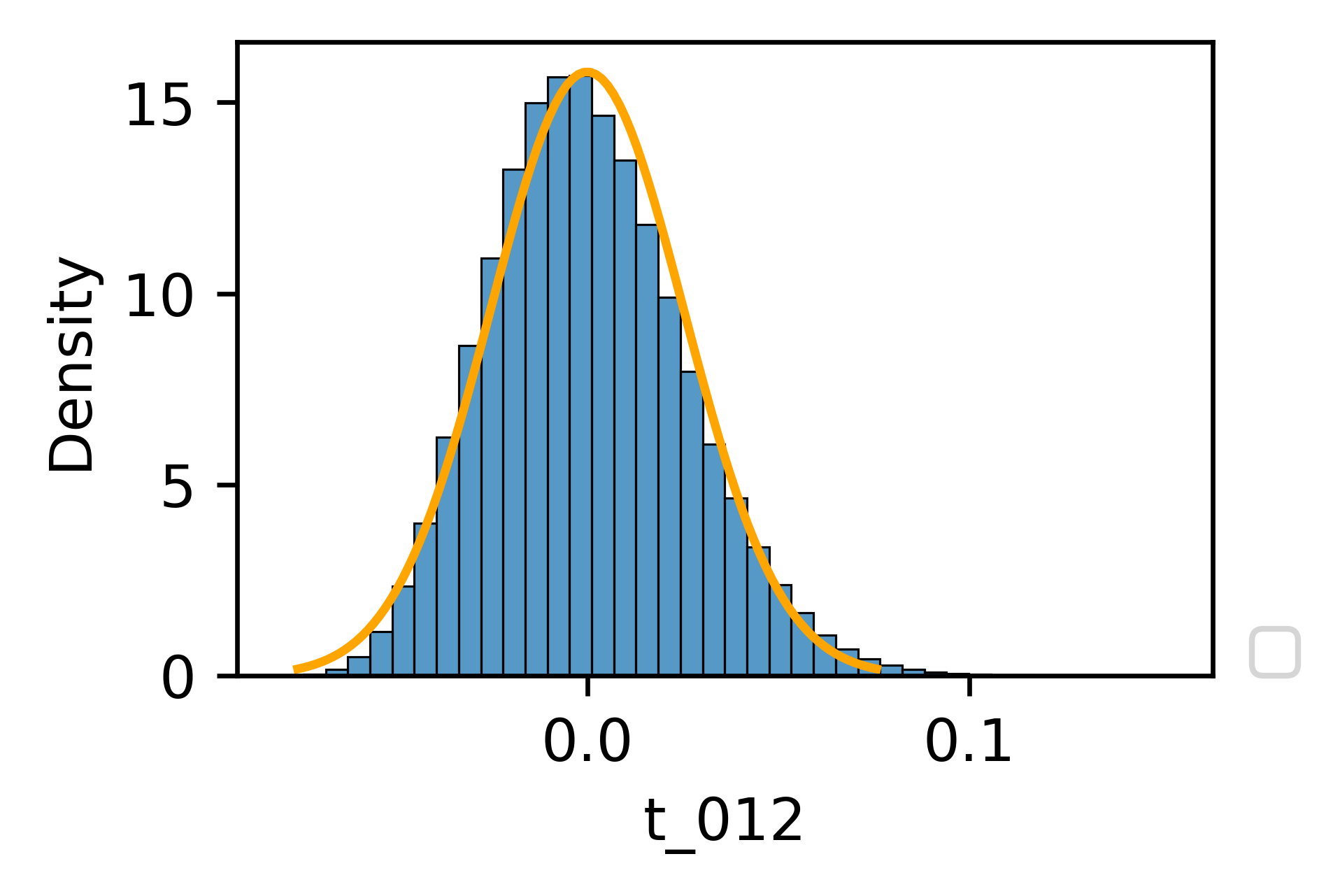}
		\caption{}
	\end{subfigure}
	\begin{subfigure}[]{0.3\textwidth}\centering
		\includegraphics[width=\textwidth]{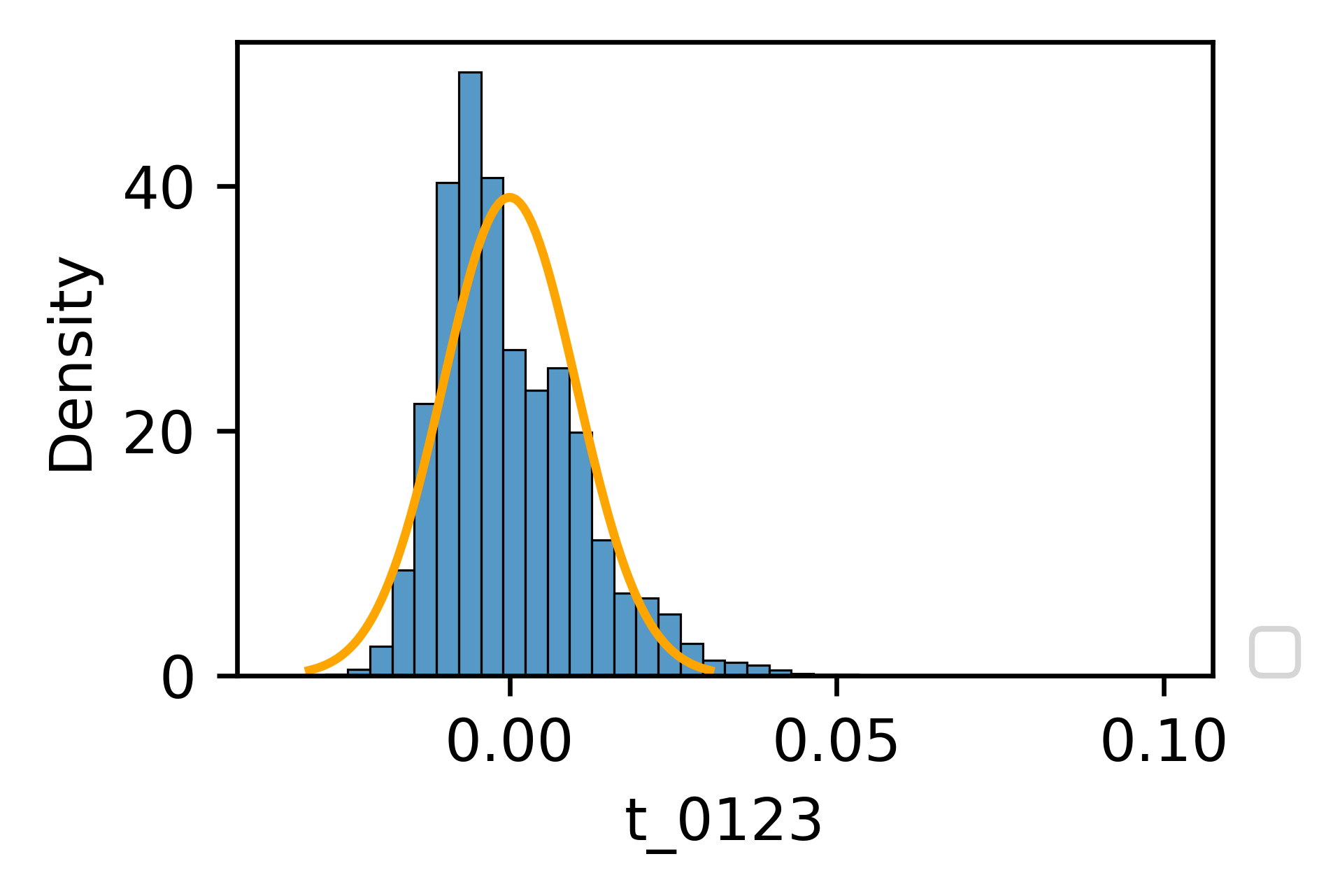}
		\caption{\label{fig-tmdh-normalp-0-1-2-3}}
	\end{subfigure}

	\begin{subfigure}[]{0.3\textwidth}\centering
		\includegraphics[width=\textwidth]{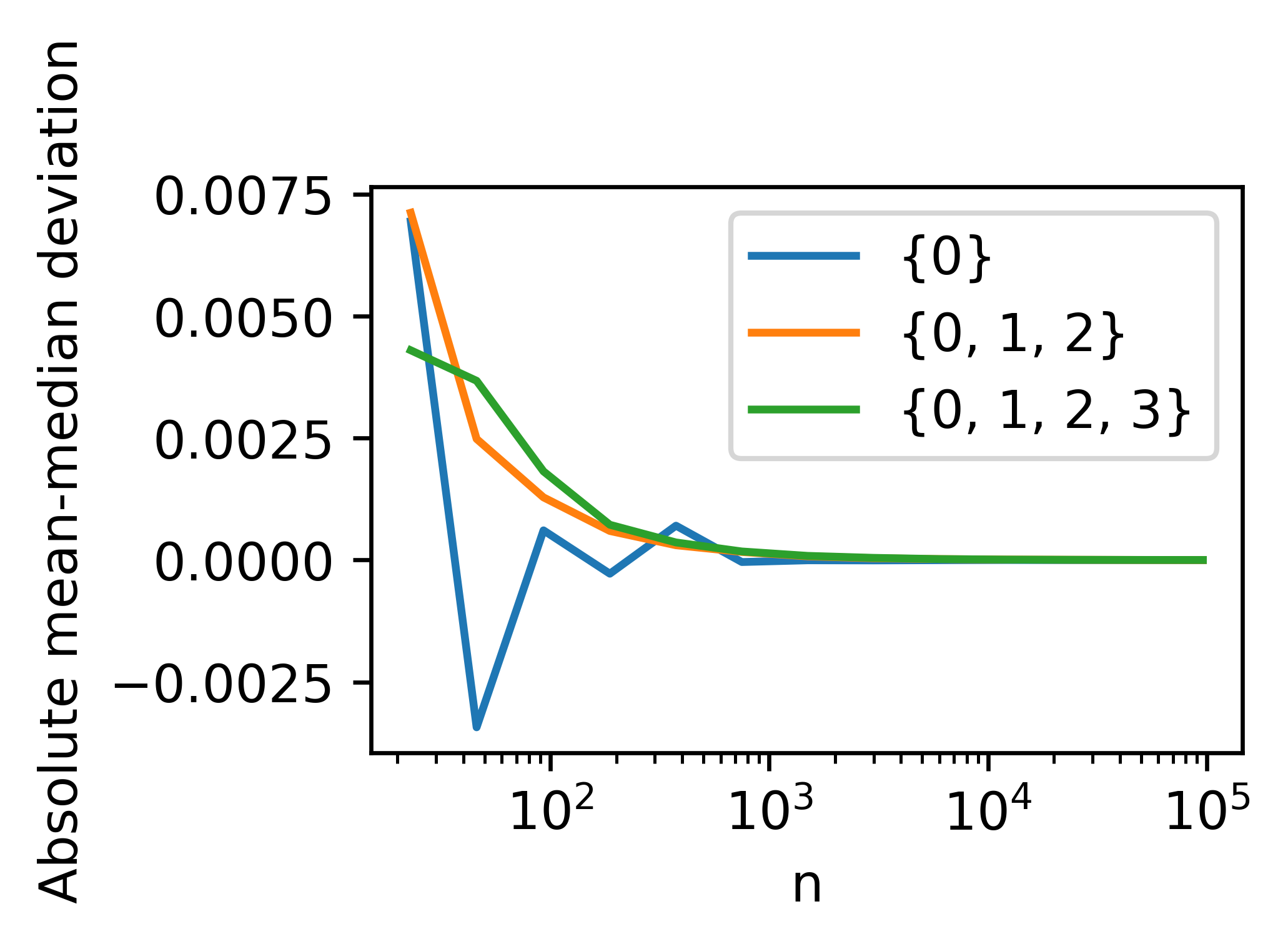}
		\caption{}
	\end{subfigure}
	\begin{subfigure}[]{0.3\textwidth}\centering
		\includegraphics[width=\textwidth]{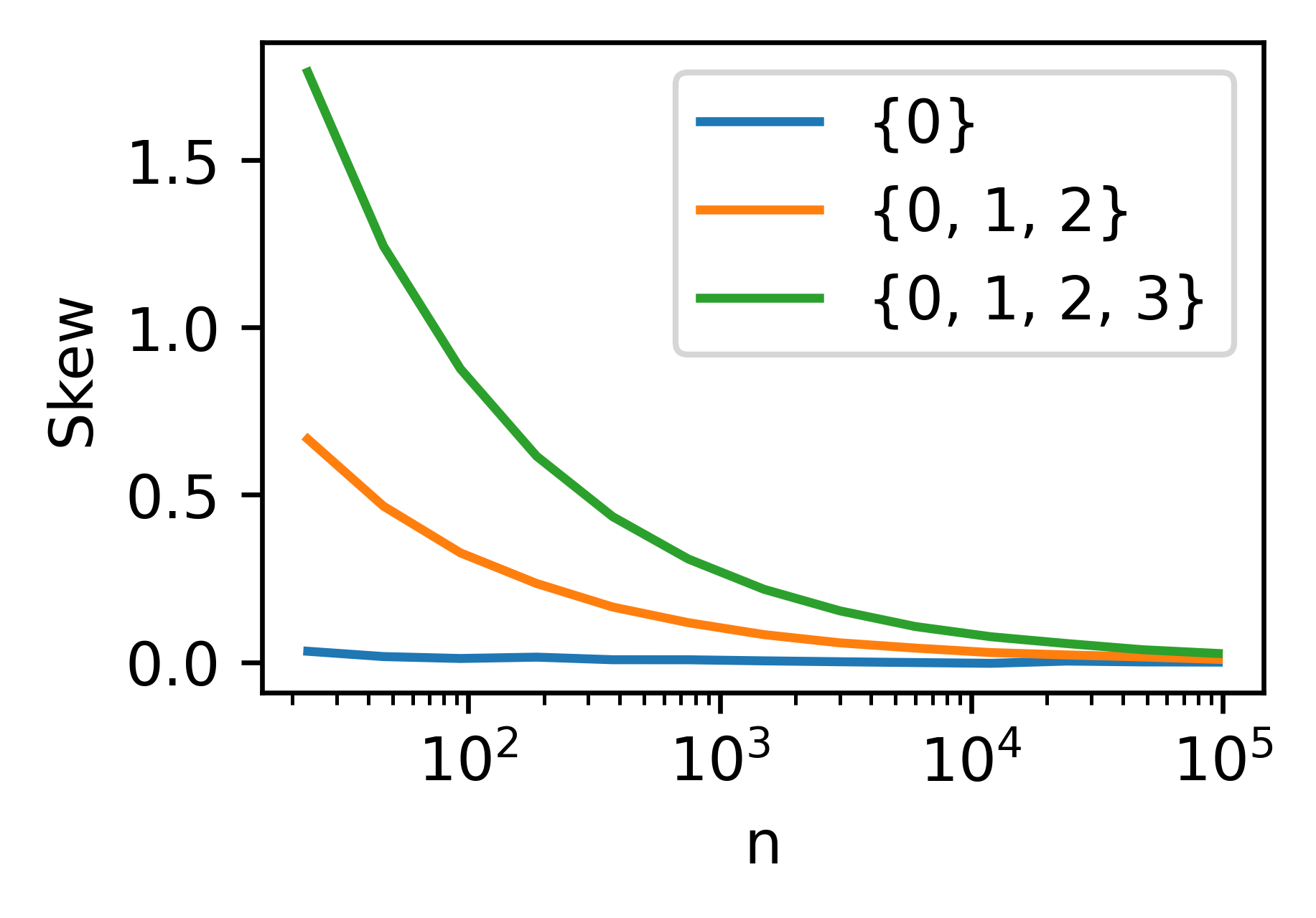}
		\caption{}
	\end{subfigure}
	\caption{Reconstruction of true frequencies based on PMDH. (a-c) Errors on normal parameters. (a-b) see text. (c) Difference between covariance of estimated means and theoretical covariance between ${1,2}$ and various sets $\mathcal{I}$ (d-f) comparison of the distribution of reconstructed $E(\tt_i)$ (blue histogram) and a normal distribution of same mean and variance (orange curve) for a dataset of 200 datapoints.
	(g-h) metrics of deviation to the normal of the distribution of reconstructed $E(\tt_i)$ : mean-median deviation (g) and skew (h).
	\label{fig-tmdh-normalp}}
\end{figure}

It would be interesting to compare the theoretical distribution of the TMDH to a numerical solution, e.g. obtained with Expectation Maximization (EM). However, our purpose here is to demonstrate that we propose a useful estimator of the joint frequencies, as well as robust error estimates on those frequencies. Therefore we consider that comparing our method with numerical approaches is beyond the scope of this paper. Another point we do not study here is the impact of priors on the estimates. Simple solutions like truncating and/or rounding the normal distribution will have either a positive or little impact on our conclusions, while more involved solutions are beyond the scope of this work.

\section{Conclusion}

In this work, we propose an normal approximation for the joint estimation of category frequencies of arbitrary dimensionality under pure LDP protocols. We beleive this can be useful both to further the research on this difficult problem as well as help speedup current numerical algorithms.

We show numerically that the normal approximation of the TMDH has a wide domain of applicability, loosing accuracy only close to the histogram boundaries. 
Future work should compare our solution to numerical solutions. We expect using the analytical TMDH could significantly improve the performance of EM-based algorithms.

We also briefly highlight potential opportunities and caveats of the normal approximation to account for prior information and constrains on the distribution support.
In section~\ref{theo-exp-c}, we discuss some constrains in the encoding space which can be handled using equation of the form~\ref{eq-sm-formula}. However, more subtle priors could be hard to handle. In principle, working with multivariate gaussian distributions offer all sorts of possibilities to build conditional probabilities. In practice however, one can easily face numerical instabilities when working directly on the posterior distribution with sharp constrains. In section~\ref{sec-val-pmdh}, we highlight how error can grow when constrains stray too far from the distribution center. Such issues should be studied more carefully to enable a wider use of the results reported here.

\appendix

\section{Decomposition of the randomized matrix}\label{app-rowtype}

We denote $\X[\mathcal{I}]$, with $\mathcal{I}$ a column set $[ij\dots]$, the slice of $\X$ constituted of the columns $i, j, \dots$ only. We denote $\X^b[\mathcal{I}]$ the
horizontal slice of $\X[\mathcal{I}]$ that corresponds to ground truth rows equal to the binary vector $b = [b_i, b_j, \dots]$ defined on $\mathcal{I}$.

\theoremstyle{definition}
\begin{definition}
A Synchronous Bernoulli Sequence (SBS) is a series of finite length (denoted $N$) of $D$ distinguishable synchronous Bernoulli experiments with fixed (permanent) success probabilities $[p_1, p_2, \dots, p_D]$ respectively.
\end{definition}

 The matrix $\X^b[\mathcal{I}]$ is analogous to a SBS where each column is a sequence, each row represent synchronous experiments, and the perturbed binary value in each cell is a single experiment output. We will denote $L(\X)$ the number of rows in a matrix $\X$.

\begin{lemma}\label{theo-row-type-counts}
For any column subset $\mathcal{I}$, and corresponding row-set defined by the bit array $b_\mathcal{J} = [b_ib_j\dots]$, denoting $\mathcal{J} = \{i | b_i = 1\}$, then
\beq
L(\X^{b_\mathcal{J}}[\mathcal{I}]) = \sum_{\mathcal{J} \subseteq \mathcal{S}\subseteq \mathcal{I}} (-1)^{|\mathcal{S}|-|\mathcal{J}|}t_\mathcal{S}
\eeq
\end{lemma}

\begin{proof}
We will prove the result using recursion.

By definition, $L(\X^{b_\mathcal{J}}[\mathcal{J}]) = t_\mathcal{J}$, so that the result is immediate when $\mathcal{J} = \mathcal{I}$.

When $|\mathcal{I}| - |\mathcal{J}| = 1$, \ie there is one and only one zero in $b_\mathcal{J}$, we have
\beq
L(\X^{b_\mathcal{J}}[\mathcal{I}]) = t_\mathcal{J} - t_\mathcal{I}
\eeq
In words, that means that the number of co-occurrences on $\mathcal{J}$ where the remaining bit of $\mathcal{I}$ is 0 is the number of co-occurrences on
$\mathcal{J}$ minus the number of co-occurrences on $\mathcal{I}$.

If we consider an arbitrary $\mathcal{J}$ with $|\mathcal{J}| < |\mathcal{I}|$, we know that we have to remove from $t_\mathcal{J}$ all the rows corresponding to bit vectors $b'$ in which at least one 0 has been flipped up compared to $b_\mathcal{J}$. We have
\beqa
L(\X^{b_\mathcal{J}}[\mathcal{I}]) &=& t_\mathcal{J} - \sum_{b'} L(\X^{b'}[\mathcal{I}]) \\
&=& t_\mathcal{J} - \sum_{\mathcal{J} \subset \mathcal{B}\subseteq \mathcal{I}}L(\X^{b'}[\mathcal{I}])
\eeqa
where $\mathcal{B} = \{i | b'_i = 1\}$. If we assume the lemma holds for all $b'$
\beqa
L(\X^{b_\mathcal{J}}[\mathcal{I}]) &=& t_\mathcal{J} -
\ \sum_{\mathcal{J} \subset \mathcal{B}\subseteq \mathcal{I}}
\ \sum_{\mathcal{B} \subseteq \mathcal{S}\subseteq \mathcal{I}} (-1)^{|\mathcal{S}|-|\mathcal{B}|}t_\mathcal{S}
\eeqa
Let us consider the contribution of a single count $t_\mathcal{S}$ in this sum. It is the same as asking, for each $l$, how many distinct $\mathcal{B}$ are subsets of $\mathcal{S}$. It is equivalent to choosing $k-l$ elements ($=|\mathcal{S} \setminus \mathcal{B}|$) from $\mathcal{S} \setminus \mathcal{I}$ to remove from $\mathcal{S}$. Since $|\mathcal{S} \setminus \mathcal{J}| = k - |\mathcal{J}|$, we have
\beqa
L(\X^{b_\mathcal{J}}[\mathcal{I}]) 
 &=& t_\mathcal{J} -
\sum_{l=|\mathcal{J}|+1}^{|\mathcal{I}|} 
\ \sum_{k=l}^{|\mathcal{I}|} (-1)^{k-l}\binom{k-|\mathcal{J}|}{k-l}
\ \sum_{\substack{\mathcal{J} \subset \mathcal{S}\subseteq \mathcal{I} \\|\mathcal{S}| = k}} t_\mathcal{S}
\\
&=& t_\mathcal{J} +
\sum_{k=|\mathcal{J}|+1}^{|\mathcal{I}|} 
\ \sum_{l=|\mathcal{J}|+1}^{k} (-1)^{k-l-1}\binom{k-|\mathcal{J}|}{k-l}
\ \sum_{\substack{\mathcal{J} \subset \mathcal{S}\subseteq \mathcal{I} \\|\mathcal{S}| = k}} t_\mathcal{S}
\\
&=& t_\mathcal{J} +
\ \sum_{k=|\mathcal{J}|+1}^{|\mathcal{I}|} (-1)^{k-|\mathcal{J}|}
\ \sum_{\substack{\mathcal{J} \subset \mathcal{S}\subseteq \mathcal{I} \\|\mathcal{S}| = k}} t_\mathcal{S}
\eeqa
where the sum on $l$ is performed using the binomial theorem. We conclude that the lemma is also valid for $b_\mathcal{J}$ with the assumption on all parent vectors $b'$. By recursion we prove the lemma.
\end{proof}

In the remainder of the appendix, we will refer to row-type and co-occurrence representations as being the histogram of distinct rows and the histogram of co-occurrences, respectively. Both form a complete representation of the joint distribution of a binary matrix.

\section{PMDH distribution}\label{app-pmdh-distri}

\subsection{Marginal distribution of co-occurrences}\label{app-marignal-c}

\begin{lemma}\label{theo-co-occ-formula}
Consider a set of columns $\mathcal{I}$ of which we count the number of co-occurrences $\cc_\mathcal{I}$ after randomization. We have
\beq
\cc_\mathcal{I} = \sum_{b=0}^{|\mathcal{I}|} \B\left(
\sum_{k=b}^{|\mathcal{I}|} (-1)^{k-b}  \binom{k}{b}
\sum_{\substack{\mathcal{S}\subseteq \mathcal{I} \\|\mathcal{S}| = k}} t_\mathcal{S},
 q^{b}p^{|\mathcal{I}|-b}\right)
\eeq
where $\B(n, p)$ signifies a Binomial random variable of parameters $n$ and $p$.
\end{lemma}
\begin{proof}
By virtue of the independence of the randomization on each row, each segment of $\X$ with the same probability of yielding a co-occurrence will contribute one Bernoulli experiment in a binomial random variable. This experiment will be positive if all positive bits remain positive (probability $=q^{|\mathcal{J}|}$) and, independently, all the remaining bits are flipped up by the randomization (probability=$p^{|\mathcal{I}|-|\mathcal{J}|}$). The total probability of a row to contribute to $\cc_\mathcal{I}$ is therefore $q^{b}p^{|\mathcal{I}|-b}$. A given $b \in [0, |\mathcal{I}|]$ corresponds to several sets $\mathcal{J}$ of positive bits in a row, all with $|\mathcal{J}| = b$. 
Using the lemma~\ref{theo-row-type-counts} for each of those yields
\beqa
\sum_{\substack{\mathcal{J} \subseteq \mathcal{I} \\|\mathcal{J}| = b}}\sum_{k=b}^{|\mathcal{I}|} (-1)^{k-b}
\sum_{\substack{\mathcal{J} \subseteq \mathcal{S}\subseteq \mathcal{I} \\|\mathcal{S}| = k}} t_\mathcal{S}
&=&
\sum_{k=b}^{|\mathcal{I}|} (-1)^{k-b}
\sum_{\substack{\mathcal{S}\subseteq \mathcal{I} \\|\mathcal{S}| = k}}t_\mathcal{S} 
\sum_{\substack{\mathcal{J} \subseteq \mathcal{S} \\|\mathcal{J}| = b}} 1\\
&=&
\sum_{k=b}^{|\mathcal{I}|} (-1)^{k-b}
\sum_{\substack{\mathcal{S}\subseteq \mathcal{I} \\|\mathcal{S}| = k}}t_\mathcal{S} 
\binom{k}{b}
\eeqa
\end{proof}

\subsection{Proof of theorem~\ref{theo-exp-c}}\label{app-exp-c}

\begin{proof}
By the linearity of the expected value, and using $E(\B(n, p)) = np$
\beqa
E(\cc_\mathcal{I}) &=& \sum_{b=0}^{|\mathcal{I}|} 
\sum_{k=b}^{|\mathcal{I}|} (-1)^{k-b}  \binom{k}{b} q^{b}p^{|\mathcal{I}|-b}
\sum_{\substack{\mathcal{S}\subseteq \mathcal{I} \\|\mathcal{S}| = k}} t_\mathcal{S}\\
 &=& \sum_{k=0}^{|\mathcal{I}|} 
\sum_{b=0}^{k} (-1)^{k-b}  \binom{k}{b} q^{b}p^{|\mathcal{I}|-b}
\sum_{\substack{\mathcal{S}\subseteq \mathcal{I} \\|\mathcal{S}| = k}} t_\mathcal{S}\\
 &=& \sum_{k=0}^{|\mathcal{I}|} 
\left(\sum_{b=0}^{k} \binom{k}{b} q^{b}(-p)^{k-b}\right)p^{|\mathcal{I}|-k}
\sum_{\substack{\mathcal{S}\subseteq \mathcal{I} \\|\mathcal{S}| = k}} t_\mathcal{S}\\
 &=& \sum_{k=0}^{|\mathcal{I}|} 
(q-p)^kp^{|\mathcal{I}|-k}
\sum_{\substack{\mathcal{S}\subseteq \mathcal{I} \\|\mathcal{S}| = k}} t_\mathcal{S}
\eeqa
\end{proof}

\subsection{Variance between arbitrary co-occurrences in SBS} \label{app-c-multinom}

\begin{lemma}\label{theo-covmat-sbs}
Consider a SBS of composed of $N$ experiments and $D = |\mathcal{I}|$ parallel sequences $\mathcal{I}$ of corresponding success probabilities $\{ p_i \ |\ i \in \mathcal{I}\}$. Consider two arbitrary column sets $\mathcal{J}$ and $\mathcal{J}'$. The covariance between the co-occurrences counts $\cc_\mathcal{J}$ and $\cc_\mathcal{J}'$ is
\beqa
V(\cc_\mathcal{J}, \cc_{\mathcal{J}'}) &=& Np_{\mathcal{J} \cup \mathcal{J}'} (1 - p_{\mathcal{J}_\cap})
\eeqa
with 
\beq
p_\mathcal{S} = \prod_{j \in \mathcal{S}} p_j
\eeq
\end{lemma}
\begin{proof}
The sum of successes and co-occurring successes in SBS can be modeled using a multinomial distribution in the row-type representation. Indeed, lemme~\ref{theo-row-type-counts} allows to compute how many rows of a given type is present in a dataset characterized by a co-occurrence histogram. Here the former representation is interpreted as a given output for set of distinguishable synchronous Bernoulli experiments (success set $\mathcal{J}$), and the latter is the partial PMDH on the corresponding SBS segment $\X^b[{\mathcal{I}}]$. We therefore define the change of distribution representation
\beq\label{eq-c-rowtype-form}
\tilde\cc_\mathcal{J} = \sum_{k=|\mathcal{J}|}^{|{\mathcal{I}}|} (-1)^{k-|\mathcal{J}|} \sum_{\substack{\mathcal{J} \subseteq \mathcal{S}\subseteq {\mathcal{I}} \\|\mathcal{S}| = k}} \cc_\mathcal{S}
\eeq
The probability that the SBS succeeds on the indexes $\mathcal{J}$ and fails elsewhere is
\beq\label{eq-rowtype-proba}
\tilde p_\mathcal{J} = \prod_{j \in \mathcal{J}} p_j \prod_{i \in {\mathcal{I}} \setminus \mathcal{J}} (1-p_i)
\eeq
Because each row can be considered as one event and that each row is independently randomized, we have that the joint probability distribution to observe any given row follows a multinomial distribution
\beq \label{eq-c-multinom}
P(\tilde\cc_\mathcal{J} = \tilde c_\mathcal{J} \ \forall\  \mathcal{J} \subseteq {\mathcal{I}}) \sim \mathrm{Mutinomial}(\{\tilde c_\mathcal{J} \ \forall\  \mathcal{J} \subseteq {\mathcal{I}}\}; \{\tilde p_\mathcal{J} \ \forall\  \mathcal{J} \subseteq {\mathcal{I}}\})
\eeq

We can then compute the full covariance matrix of $\mathbf{\tilde C}= \{\tilde \cc_\mathcal{J} \ \forall\  \mathcal{J} \subseteq {\mathcal{I}}\}$ with
\beq
V(\mathbf{\tilde C}) = N \left(\mathrm{diag}(\tilde P) - \tilde P\tilde P^T \right)
\eeq
where $\tilde P = \{\tilde p_\mathcal{J} \ \forall\  \mathcal{J} \subseteq {\mathcal{I}}\}$.
The invert transformation from the row-type form back to the co-occurrence form is
\beq\label{eq-c-inverse-rowtype-form}
\cc_\mathcal{J} = \sum_{\mathcal{J} \subseteq \mathcal{S}\subseteq {\mathcal{I}}} \tilde\cc_\mathcal{S}
\eeq
so that by linearity of the covariance operator we can write
\beqa
V(\cc_\mathcal{J}, \cc_{\mathcal{J}'}) 
&=& 
\sum_{\substack{\mathcal{J} \subseteq \mathcal{S}\subseteq {\mathcal{I}} \\
	            \mathcal{J}' \subseteq \mathcal{S}'\subseteq {\mathcal{I}}}}
V(\tilde\cc_\mathcal{S}, \tilde\cc_{\mathcal{S}'})
\eeqa
which become after substituting with the result of lemma~\ref{theo-covmat-sbs}
\beqa
V(\cc_\mathcal{J}, \cc_{\mathcal{J}'}) {}
&=& 
N
\sum_{\substack{\mathcal{J} \subseteq \mathcal{S}\subseteq {\mathcal{I}} \\
	            \mathcal{J}' \subseteq \mathcal{S}'\subseteq {\mathcal{I}}}}
	\tilde p_\mathcal{S} \left(\delta(\mathcal{S}, \mathcal{S}') -  \tilde p_{\mathcal{S}'}
\right)
\eeqa
The probabilities~\ref{eq-rowtype-proba} behave as a vector in the row-type to co-occurrence transformation~\ref{eq-c-inverse-rowtype-form}
\beq\label{eq-proba-inverse-rowtype-form}
p_\mathcal{J} = \sum_{\mathcal{J} \subseteq \mathcal{S}\subseteq {\mathcal{I}}} \tilde p_\mathcal{S}
\eeq
This can be understood easily by the fact that the sum of all $\tilde p_\mathcal{S}/p_\mathcal{J}$, representing the probability of any variation of columns outside of $\mathcal{J}$, sum to 1. We obtain
\beqa
V(\cc_\mathcal{J}, \cc_{\mathcal{J}'}) /N
&=& 
- p_{\mathcal{J}}p_{\mathcal{J}'} +
\sum_{\mathcal{J}\cup \mathcal{J}'\subseteq \mathcal{S}\subseteq {\mathcal{I}}}
	\tilde p_\mathcal{S}\\
\eeqa
The final result is a consequence of equation~\ref{eq-proba-inverse-rowtype-form}  and the event independence
\beq
p_{\mathcal{J}}p_{\mathcal{J}'} = p_{\mathcal{J \cup J}'} (1- p_{\mathcal{J \cap J}'}).
\eeq
\end{proof}

\subsection{Covariance matrix of PMDH}
\subsubsection{Proof of theorem~\ref{theo-var-c}}\label{app-var-c}
\begin{proof}
By the independence of each row in $\X$, the linearity of the variance, and using $V(\B(n, p)) = np(1-p)$ we have
\beqa
V(\cc_\mathcal{I}) &=& \sum_{b=0}^{|\mathcal{I}|} 
\sum_{k=b}^{|\mathcal{I}|} (-1)^{k-b}  \binom{k}{b} q^{b}p^{|\mathcal{I}|-b}(1-q^{b}p^{|\mathcal{I}|-b})
\sum_{\substack{\mathcal{S}\subseteq \mathcal{I} \\|\mathcal{S}| = k}} t_\mathcal{S}\\
 &=& \sum_{k=0}^{|\mathcal{I}|} 
\sum_{b=0}^{k} \left(\binom{k}{b} q^{b}(-p)^{k-b}p^{|\mathcal{I}|-k} -  q^{2b}(-p^2)^{k-b}p^{2(|\mathcal{I}|-k)} \right)
\sum_{\substack{\mathcal{S}\subseteq \mathcal{I} \\|\mathcal{S}| = k}} t_\mathcal{S}\\
 &=& \sum_{k=0}^{|\mathcal{I}|} 
(q-p)^kp^{|\mathcal{I}|-k}(1 - (q+p)^kp^{|\mathcal{I}|-k})
\sum_{\substack{\mathcal{S}\subseteq \mathcal{I} \\|\mathcal{S}| = k}} t_\mathcal{S}
\eeqa
which leads directly to the result.
\end{proof}

\subsubsection{Proof of theorem~\ref{theo-covar-c}}\label{app-covar-c}
\begin{proof}
For convenience, we introduce the notation $\mathcal{J}_\cup = \mathcal{J}\cup \mathcal{J}'$.
Theorem~\ref{theo-covmat-sbs} provides an expression for the covariance matrix in a segment $\X^b[{\mathcal{J}_\cup}]$ of $\X$. 
By the independence of each row in $\X$, and thanks to the covariance linearity, we have 
\beq
V(\cc_\mathcal{J}, \cc_{\mathcal{J}'})  = \sum_{\mathcal{B} \subseteq \mathcal{J}_\cup} V(\cc^{(b)}_\mathcal{J}, \cc^{(b)}_{\mathcal{J}'})
\eeq
where the sum is performed on all possible $b$ ($\sim \mathcal{B}$) and $\cc^{(b)}$ refers to the partial PMDH computed on the segment defined by $b$. We have using $L$ to count the number of elements in $\X^b[{\mathcal{J}_\cup}]$
\beqa
V(\cc_\mathcal{J}, \cc_{\mathcal{J}'})  &=& \sum_{\mathcal{B} \subseteq \mathcal{J}_\cup}  L(\X^{b}[{\mathcal{J}_\cup}]) \ 
p_{\mathcal{J} \cup \mathcal{J}'} (1 - p_{\mathcal{J}_\cap})
\eeqa
We pause to interpret the above equation. Choosing $b$ is equivalent to choosing a set $\mathcal{B}$ defining the ground truth from which a row has been generated. $\mathcal{J}$ and $\mathcal{J}'$ however represent the \emph{output} of the randomization. They define the actual row in $\X$, and $\tilde p_\mathcal{J}$ is the probability to obtain it from a ground truth $b$. Therefore $\tilde p_\mathcal{J}$ depends implicitly on $b$
\beq
p_\mathcal{J}(\mathcal{B}) = 
q^{|\mathcal{J} \cap \mathcal{B}|} 
p^{|\mathcal{J} \setminus \mathcal{B}|} 
\eeq
We use lemma~\ref{theo-row-type-counts} and introduce the shorthand notation $\mathcal{J}_\cap = \mathcal{J} \cap \mathcal{J}'$ to obtain
\beqa
V(\cc_\mathcal{J}, \cc_{\mathcal{J}'})  &=& 
\sum_{\mathcal{B}\subseteq {\mathcal{J}_\cup}} 
p_{\mathcal{J}_\cup}(\mathcal{B}) (1 - p_{\mathcal{J}_\cap}(\mathcal{B}))
\sum_{\mathcal{B} \subseteq \mathcal{S}\subseteq {\mathcal{J}_\cup}}
(-)^{|\mathcal{S}| - |\mathcal{B}|}
t_{\mathcal{S}}\\
 &=& 
\sum_{\mathcal{B}\subseteq {\mathcal{J}_\cup}} 
q^{|\mathcal{B}|} p^{|\mathcal{J}_\cup| - |\mathcal{B}|} (1 - p_{\mathcal{J}_\cap}(\mathcal{B}))
\sum_{\mathcal{B} \subseteq \mathcal{S}\subseteq {\mathcal{J}_\cup}}
(-)^{|\mathcal{S}| - |\mathcal{B}|}
t_{\mathcal{S}}\\
 &=& 
\sum_{\mathcal{S}\subseteq {\mathcal{J}_\cup}}
t_{\mathcal{S}}\ p^{|\mathcal{J}_\cup| - |\mathcal{S}|}
\sum_{\mathcal{B}\subseteq \mathcal{S}} 
q^{|\mathcal{B}|} (-p)^{|\mathcal{S}| - |\mathcal{B}|} (1 - p_{\mathcal{J}_\cap}(\mathcal{B}))\\
 &=& 
\sum_{\mathcal{S}\subseteq {\mathcal{J}_\cup}}
t_{\mathcal{S}} \ p^{|\mathcal{J}_\cup| - |\mathcal{S}|}
\left(
(q-p)^{|\mathcal{S}|} -
\sum_{\mathcal{B}\subseteq \mathcal{S}}
q^{|\mathcal{B}| + |\mathcal{J}_\cap \cap \mathcal{B}|} (-p)^{|\mathcal{S}| - |\mathcal{B}|} p^{|\mathcal{J}_\cap \setminus \mathcal{B}|}
\right)\\
 &=& 
\sum_{\mathcal{S}\subseteq {\mathcal{J}_\cup}}
t_{\mathcal{S}} \ p^{|\mathcal{J}_\cup| - |\mathcal{S}|}
\left(
(q-p)^{|\mathcal{S}|}
-\sum_{\mathcal{B}'\subseteq (\mathcal{S} \cap \mathcal{J}_\cap)}
q^{|\mathcal{B}'|} p^{|\mathcal{J}_\cap| - |\mathcal{B}'|}
\sum_{\mathcal{B}''\subseteq (\mathcal{S} \setminus \mathcal{J}_\cap)}
q^{|\mathcal{B}|} (-p)^{|\mathcal{S}| - |\mathcal{B}|}
\right)
\eeqa
where in the last step we split $\mathcal{B}$ into the part intersecting $\mathcal{B}$, $\mathcal{B}'$ and the part excluding $\mathcal{B}$, $\mathcal{B}''$. We have $\mathcal{B} = \mathcal{B}' + \mathcal{B}''$ and $|\mathcal{B}| = |\mathcal{B}'| + |\mathcal{B}''|$. Using the binomial theorem twice, we obtain the result
\beqa
V(\cc_\mathcal{J}, \cc_{\mathcal{J}'})  
 &=& 
\sum_{\mathcal{S}\subseteq {\mathcal{J}_\cup}}
t_{\mathcal{S}} \ p^{|\mathcal{J}_\cup| - |\mathcal{S}|}\nonumber\\ && 
\left(
(q-p)^{|\mathcal{S}|}
-
p^{|\mathcal{J}_\cap| - |\mathcal{S} \cap \mathcal{J}_\cap|}
\sum_{\mathcal{B}'\subseteq (\mathcal{S} \cap \mathcal{J}_\cap)}
q^{2|\mathcal{B}'|} (-p^2)^{|\mathcal{S} \cap \mathcal{J}_\cap| - |\mathcal{B}'|}
\sum_{\mathcal{B}''\subseteq (\mathcal{S} \setminus \mathcal{J}_\cap)}
q^{|\mathcal{B}''|} (-p)^{|\mathcal{S} \setminus \mathcal{J}_\cap| - |\mathcal{B}''|}
\right)\nonumber\\\\
 &=& 
\sum_{\mathcal{S}\subseteq {\mathcal{J}_\cup}}
t_{\mathcal{S}} \ p^{|\mathcal{J}_\cup| - |\mathcal{S}|}
\left(
(q-p)^{|\mathcal{S}|}
-
p^{|\mathcal{J}_\cap| - |\mathcal{S} \cap \mathcal{J}_\cap|}
(q^2-p^2)^{|\mathcal{S} \cap \mathcal{J}_\cap|}
(q-p)^{|\mathcal{S} \setminus \mathcal{J}_\cap|}
\right)\\
 &=& 
\sum_{\mathcal{S}\subseteq {\mathcal{J}_\cup}}
t_{\mathcal{S}} \ p^{|\mathcal{J}_\cup| - |\mathcal{S}|}\ (q-p)^{|\mathcal{S}|}
\left(
1
-
p^{|\mathcal{J}_\cap| - |\mathcal{S} \cap \mathcal{J}_\cap|}
(q+p)^{|\mathcal{S} \cap \mathcal{J}_\cap|}
\right)
\eeqa
\end{proof}

\subsection{Example of M matrix of PMDH for 3 columns and it's inverse}\label{app-ex-m3}

For three columns, the maximum co-occurrence order is 3. The first block is composed of the three column counts $\cc_1, \cc_2,$ and $\cc_3$. The second block is also three-dimensional, with $\cc_{12}, \cc_{13}$, and $\cc_{23}$. The last block has only one dimension, $\cc_{123}$
We then have
\beq
M  =
\begin{bmatrix}
(q-p) \id & 0 & 0 \\
p(q-p)\km^{(21)} & (q-p)^2 \id & 0 \\
p^2(q-p)\ \km^{(31)}  & p(q-p)^2\ \km^{(32)}  & (q-p)^3 \\
\end{bmatrix} 
\eeq
where here we denote the 3 by 3 identity matrix $\id$ for simplicity. We also have
\beq
\km^{(21)} = 
\begin{bmatrix}
1 & 1 & 0 \\
1 & 0 & 1 \\
0 & 1 & 1
\end{bmatrix}\qquad \rm{and} \qquad \km^{(31)} = \km^{(32)} = [1, 1, 1].
\eeq

The inverse $M^{-1}$ is

\beq
M^{-1}  =
\frac{1}{(q-p)^3}
\left(
\begin{bmatrix}
(q-p)^2 \id & 0 & 0 \\
-p(q-p)\km^{(21)}  & (q-p) \id & 0 \\
p^2\ \km^{(31)} & -p\ \km^{(32)} & 1 \\
\end{bmatrix} 
\right)
\eeq

As an example, we can verify that the  $(31)$ element of the product $MM^{-1}$ vanishes
\beqa
p^2(q-p)\ \km^{(31)} (q-p)^{-1} \id - p(q-p)^2\ \km^{(32)} p(q-p)^{-2}\km^{(21)} + (q-p)^3 p^2(q-p)^{-3}\ \km^{(31)} &=& 0 \\
p^2\ \km^{(31)} - p^2\ \km^{(32)}\km^{(21)} + p^2\ \km^{(31)} &=& 0 
\eeqa
which is obtained because we have
\beq
\km^{(32)}\km^{(21)} = 2\km^{(31)}
\eeq
as demonstrated in the next section (lemma~\ref{theo-km-prod}).

\section{TMDH distribution}

\subsection{Proof of theorem~\ref{theo-inv-m}}\label{app-inv-m}

\begin{lemma}\label{theo-km-prod}
For any $\alpha > \beta > \gamma$, we have
\beq
\km^{(\alpha\beta)}\km^{(\beta\gamma)} = \binom{\alpha-\gamma}{\alpha-\beta} \km^{(\alpha\gamma)}
\eeq
\end{lemma}
\begin{proof}
We introduce the arbitrary sets $\mathcal{A}$ and $\mathcal{C}$, of size $\alpha$ and $\gamma$ respectively. We have
\beqa
(\km^{(\alpha\beta)}\km^{(\beta\gamma)})(\mathcal{A}, \mathcal{C}) 
&=& \sum_{\mathcal{B}, |\mathcal{B}| = \beta} \km^{(\alpha\beta)}(\mathcal{A}, \mathcal{B})\km^{(\beta\gamma)}(\mathcal{B}, \mathcal{C})\\
&=& \sum_{\substack{\mathcal{C} \subset \mathcal{B} \subset \mathcal{A}\\ |\mathcal{B}| = \beta}} 1 \\
&=& \binom{\alpha-\gamma}{\alpha-\beta}\km^{(\alpha\gamma)}(\mathcal{A}, \mathcal{C}) 
\eeqa
\end{proof}
We now turn to the proof of theorem~\ref{theo-inv-m}.
\begin{proof}
In order to obtain the TMDH normal parameters, we first need to invert the matrix $M$. This can be done using the Neumann series:
\beq\label{eq-neumann}
(A-B)^{-1} = (\id  - A^{-1}B)^{-1}A^{-1}  = \left(\sum_{k = 0}^\infty (A^{-1}B)^k\right) A^{-1} 
\eeq
In our case, we have
\beq
A^{-1} =  \sum_{\alpha=1}^{D} (q-p)^{-\alpha} \id^{(\alpha)}
\eeq
and 
\beq
A^{-1}B = -\sum_{\alpha=2}^{D}\sum_{\beta = 1}^{\alpha-1} \left(\frac{p}{q-p}\right)^{\alpha-\beta} \km^{(\alpha\beta)}
\eeq
This is a lower triangular matrix so that we are assured that the series~\ref{eq-neumann} terminates. Using the lemma~\ref{theo-km-prod} we obtain
\beqa
\left(\sum_{\alpha=2}^{D}\sum_{\beta = 1}^{\alpha-1} \left(\frac{p}{q-p}\right)^{\alpha-\beta} \km^{(\alpha\beta)}\right)^2
&=&
\sum_{\alpha=3}^{D}\sum_{\beta = 2}^{\alpha-1}\sum_{\gamma = 1}^{\beta-1}
\left(\frac{p}{q-p}\right)^{\alpha - \gamma}  \binom{\alpha-\gamma}{\alpha-\beta}\km^{(\alpha\gamma)}\\
&=&
\sum_{\alpha=3}^{D}\sum_{\gamma = 1}^{\alpha-2}\left(\frac{p}{q-p}\right)^{\alpha - \gamma}\km^{(\alpha\gamma)}
\sum_{\beta' = 1}^{\alpha-\gamma-1}\binom{\alpha-\gamma}{\beta'}\\
&=&
\sum_{\alpha=3}^{D}\sum_{\beta = 1}^{\alpha-2}(2^{\alpha-\beta} - 2)\left(\frac{p}{q-p}\right)^{\alpha - \beta}\km^{(\alpha\beta)}
\eeqa
where we rename $\gamma$ to $\beta$ in the last step. Similarly
\beqa
\left(\sum_{\alpha=2}^{D}\sum_{\beta = 1}^{\alpha-1} \left(\frac{p}{q-p}\right)^{\alpha-\beta} \km^{(\alpha\beta)}\right)^3
&=&
\sum_{\alpha=4}^{D}\sum_{\beta = 2}^{\alpha-2}\sum_{\gamma = 1}^{\beta-1}(2^{\alpha-\beta} - 2)\left(\frac{p}{q-p}\right)^{\alpha - \gamma}
\binom{\alpha-\gamma}{\alpha-\beta}\km^{(\alpha\gamma)}\\
&=&
\sum_{\alpha=4}^{D}\sum_{\gamma = 1}^{\alpha-3}\left(\frac{p}{q-p}\right)^{\alpha - \gamma}\km^{(\alpha\gamma)}
\sum_{\beta' = 1}^{\alpha-\gamma-2}(2^{\alpha-\gamma-\beta'} - 2)\binom{\alpha-\gamma}{\beta'}
\eeqa
By repeating the process we deduce that in general
\beq\label{eq-A-1Bk}
\left(\sum_{\alpha=2}^{D}\sum_{\beta = 1}^{\alpha-1} \left(\frac{p}{q-p}\right)^{\alpha-\beta} \km^{(\alpha\beta)}\right)^k
=
\sum_{\alpha=k+1}^{D}\sum_{\gamma = 1}^{\alpha-k}
\kappa^{(k)}_{\alpha-\gamma}\left(\frac{p}{q-p}\right)^{\alpha - \gamma}\km^{(\alpha\gamma)}
\eeq
where $\kappa^{(k)}_{\alpha\gamma}$ are combinatorial coefficients that can be computed recursively
\beq
\kappa^{(k+1)}_{\delta} 
= \sum_{\beta = 1}^{\delta-k}\kappa^{(k)}_{\delta - \beta}\binom{\delta}{\beta}
= \sum_{\beta' = k}^{\delta-1}\kappa^{(k)}_{\beta'}\binom{\delta}{\beta}
\eeq
Therefore $\kappa^{(k)}_{\delta} = 0$ if $k>\delta$ and the highest non-vanishing $k$ is $D-1$. 
Let us investigate those coefficients with combinatorics reasoning.

The coefficients $\kappa^{(k)}_{\delta}$ appear also in the exponentiation of the full binary matrix $\km$ composed of all the blocks $\km^{(\alpha\gamma)}$ without factors. We can interpret $\km^k$ as the matrix counting all the ways to connect any to subset of the range $[0, D]$ using $k$ steps. If we consider $\km^k(\mathcal{A}, \mathcal{B})$, then the corresponding weight is $\kappa^{(k)}_{\alpha-\beta}$ if $\mathcal{B} \subset \mathcal{A}$. This count represents all the ways $\mathcal{A} \setminus \mathcal{B}$ can be divided into $k$ partitions in a given order. Bare the ordering, this corresponds to the Stirling numbers of the second kind $S(n, k)$ defined as
\beq
S(n, k) = \frac{1}{k!}\sum_{i=0}^k (-)^i \binom{k}{i}(k-i)^n
\eeq
In our case, we have to multiply those factors to account for the permutations of the partitions
\beq
\kappa^{(k)}_{\delta} = k! S(\delta, k)
\eeq
Those numbers are listed as the Triangle of numbers in ref.~\cite{A019538}. They obey the triangular relation
\beq\label{eq-triang-num-triang}
\kappa^{(k)}_{\delta} = k \left( \kappa^{(k-1)}_{\delta-1} + \kappa^{(k)}_{\delta-1}\right)
\eeq

Substituting equation~\ref{eq-A-1Bk} in the series~\ref{eq-neumann} and replacing $\gamma$ by $\beta$ we have
\beqa
M^{-1} 
&=&  \left(\sum_{k = 0}^{D-1} (-)^k \sum_{\alpha=k+1}^{D}\sum_{\beta = 1}^{\alpha-k}
\kappa^{(k)}_{\alpha-\beta}\left(\frac{p}{q-p}\right)^{\alpha - \beta}\km^{(\alpha\beta)}\right)
A^{-1}\\
&=& 
 \sum_{\alpha=1}^{D} (q-p)^{-\alpha} 
 \left( \id^{(\alpha)} +  \sum_{\beta=1}^{\alpha-1}  \vkappa_{\alpha-\beta}\ p^{\alpha - \beta}\ \km^{(\alpha\beta)}  \right)
\eeqa
where
\beqa
\vkappa_{\delta} 
&=& 
\sum_{k = 1}^{\delta} (-)^k\kappa^{(k)}_{\delta} \\
\eeqa
We can verify manually that the first $\vkappa_{\delta}$ are equal to $(-)^\delta$. Using~the triangle relation~\ref{eq-triang-num-triang} we have
\beqa
\vkappa_{\delta} 
&=& 
\sum_{k = 1}^{\delta} (-)^k k \left( \kappa^{(k-1)}_{\delta-1} + \kappa^{(k)}_{\delta-1}\right) \\
&=& 
(-)^\delta\delta \kappa^{\delta}_{\delta-1} + \kappa^{0}_{\delta-1} + \sum_{k = 1}^{\delta-1} (-)^k \left( -(k+1) \kappa^{(k)}_{\delta-1} + k \kappa^{(k)}_{\delta-1}\right) \\
&=& 
-\vkappa_{\delta-1} 
\eeqa
with $\kappa^{\delta}_{\delta-1} = \kappa^{0}_{\delta-1} = 0$

\end{proof}

\subsection{Proof of theorems~\ref{theo-t-exp} and~\ref{theo-t-var}}\label{app-t-params}

\begin{proof}[proof of theorem~\ref{theo-t-exp}]
\beqa
\mu_\mathcal{I}(C) &=& \left(M^{-1} (C - NP)\right)_{\mathcal{I}} \\
&=&  \sum_\mathcal{S} M_{\mathcal{I}\mathcal{S}}^{-1} \left(c_\mathcal{S} - p^{|\mathcal{S}|}N\right) \\
&=&  \sum_\mathcal{S} (q-p)^{-|\mathcal{I}|} (-p)^{|\mathcal{I}| - |\mathcal{S}|} \km^{(|\mathcal{I}||\mathcal{S}|)}(\mathcal{I}, \mathcal{S}) \left(c_\mathcal{S} - p^{|\mathcal{S}|}N\right) \\
&=&  \sum_{\mathcal{S} \subseteq \mathcal{I}} \frac{(-p)^{|\mathcal{I}| - |\mathcal{S}|}}{(q-p)^{|\mathcal{I}|}} \left(c_\mathcal{S} - p^{|\mathcal{S}|}N\right)\\
&=&
({q-p})^{-|\mathcal{I}|}
\sum_{\mathcal{S} \subseteq \mathcal{I}} (-p)^{|\mathcal{I}| - |\mathcal{S}|}c_\mathcal{S}
\eeqa
where we use the binomial formula together with $\cc_\emptyset = N$ in the last step.
\end{proof}

\begin{proof}[proof of theorem~\ref{theo-t-var}]
For any matrix element of $\hat \Sigma(\mu)$ we have using theorems~\ref{theo-covar-c} and~\ref{theo-t-exp}
\beqa
V(\cc_\mathcal{J}(\mu), \cc_{\mathcal{J}'}(\mu)) 
&=& 
\sum_{\mathcal{S}\subseteq {\mathcal{J}_\cup}}
\mu_{\mathcal{S}} \ (q-p)^{|\mathcal{S}|}\ p^{|\mathcal{J}_\cup| - |\mathcal{S}|}
\left(1 - (q+p)^{|\mathcal{S} \cap \mathcal{J}_\cap|} p^{|\mathcal{J}_\cap| - |\mathcal{S} \cap \mathcal{J}_\cap|} \right)\\
&=& 
\sum_{\mathcal{S'}\subseteq\mathcal{S}\subseteq {\mathcal{J}_\cup}} c_{\mathcal{S'}}
 (-)^{|\mathcal{S}| - |\mathcal{S'}|} 
 p^{|\mathcal{J}_\cup| - |\mathcal{S'}|}
\left(1 - (q+p)^{|\mathcal{S} \cap \mathcal{J}_\cap|} p^{|\mathcal{J}_\cap| - |\mathcal{S} \cap \mathcal{J}_\cap|} \right)
\eeqa
Each term depends on $\mathcal{S}$ only through the sign $(-)^{|\mathcal{S}|}$ and the size of its overlap with the intersection $\mathcal{J}_\cap$ ($|\mathcal{S} \cap \mathcal{J}_\cap|$). The summation on $\mathcal{S}$ can be simplified using $\mathcal{J}_\Delta = \mathcal{J}_\cup \setminus \mathcal{J}_\cap$
\beqa
&&
\sum_{s_\cap=0}^{|\mathcal{J}_\cap \setminus \mathcal{S'}|}
\sum_{s_\Delta=0}^{|\mathcal{J}_\Delta \setminus \mathcal{S'}|}
 (-)^{|\mathcal{S'}|+s_\cap+s_\Delta}
 \binom{|\mathcal{J}_\cap \setminus \mathcal{S'}|}{s_\cap}
 \binom{|\mathcal{J}_\Delta \setminus \mathcal{S'}|}{s_\Delta}
 \left(1 - (q+p)^{|\mathcal{S'} \cap \mathcal{J}_\cap| + s_\cap} p^{|\mathcal{J}_\cap| - |\mathcal{S'} \cap \mathcal{J}_\cap| - s_\cap} \right)\nonumber\\\\
 &=&
 (-)^{|\mathcal{S'}|}
\id(\mathcal{J}_\Delta \setminus \mathcal{S'}, \emptyset)
\sum_{s_\cap=0}^{|\mathcal{J}_\cap \setminus \mathcal{S'}|}
 (-)^{s_\cap}
 \binom{|\mathcal{J}_\cap \setminus \mathcal{S'}|}{s_\cap}
 \left(1 - (q+p)^{|\mathcal{S'} \cap \mathcal{J}_\cap| + s_\cap} p^{|\mathcal{J}_\cap \setminus \mathcal{S'}| - s_\cap} \right)\\
 &=&
 (-)^{|\mathcal{S'}|}
\id(\mathcal{J}_\Delta \setminus \mathcal{S'}, \emptyset)
 \left(\id(\mathcal{J}_\cap \setminus \mathcal{S'}, \emptyset)
- (q+p)^{|\mathcal{S'} \cap \mathcal{J}_\cap|} (-q)^{|\mathcal{J}_\cap \setminus \mathcal{S'}|} \right)
\eeqa
so that
\beqa
V(\cc_\mathcal{J}(\mu), \cc_{\mathcal{J}'}(\mu)) 
&=& 
\sum_{\mathcal{S'}\subseteq{\mathcal{J}_\cup}} c_{\mathcal{S'}}
p^{|\mathcal{J}_\cup| - |\mathcal{S'}|}
\id(\mathcal{J}_\Delta \setminus \mathcal{S'}, \emptyset)
 \left(\id(\mathcal{J}_\cap \setminus \mathcal{S'}, \emptyset)
- (q+p)^{|\mathcal{S'} \cap \mathcal{J}_\cap|} (-q)^{|\mathcal{J}_\cap \setminus \mathcal{S'}|} \right)\nonumber\\\\
&=& 
\sum_{\mathcal{J}_\Delta \subseteq \mathcal{S'}\subseteq{\mathcal{J}_\cup}}
 c_{\mathcal{S'}}
p^{|\mathcal{J}_\cup| - |\mathcal{S'}|}
 \left(\id(\mathcal{J}_\cap \setminus \mathcal{S'}, \emptyset)
- (q+p)^{|\mathcal{S'} \cap \mathcal{J}_\cap|} (-q)^{|\mathcal{J}_\cap \setminus \mathcal{S'}|} \right)
\eeqa

We can then develop equation~\ref{eq-tot-var} and inject the identity deduced above
\beqa
V(\tt_\mathcal{J}, \tt_{\mathcal{J}'}) 
&=&
(q-p)^{-(|\mathcal{J}|+|\mathcal{J}'|)} 
\sum_{\substack{\mathcal{B} \subseteq \mathcal{J} \\ \mathcal{B}' \subseteq \mathcal{J}'}} 
(-p)^{|\mathcal{J}| + |\mathcal{J}'| - |\mathcal{B}| - |\mathcal{B}'|}
V(\cc_{\mathcal{B}}, \cc_{\mathcal{B}'})\\
&=&
\left(\frac{-p}{q-p}\right)^{|\mathcal{J}|+|\mathcal{J}'|}\!\!
\sum_{\substack{\mathcal{B} \subseteq \mathcal{J} \\ \mathcal{B}' \subseteq \mathcal{J}'\\\mathcal{B}_\Delta \subseteq \mathcal{S'} \subseteq \mathcal{B}_\cup} }
 c_{\mathcal{S'}}\
(-)^{|\mathcal{B}|+|\mathcal{B'}|}
\frac{
 \id(\mathcal{B}_\cap \setminus \mathcal{S'}, \emptyset)
- (q+p)^{|\mathcal{B}_\cap \cap \mathcal{S'}|} (-q)^{|\mathcal{B}_\cap \setminus \mathcal{S'}|} }
{p^{2|\mathcal{B}_\cup| - |\mathcal{B}_\Delta| - |\mathcal{B}_\cup| + |\mathcal{S'}|}}\\
 &=&
\left(\frac{-p}{q-p}\right)^{|\mathcal{J}|+|\mathcal{J}'|}\!\!
\sum_{\mathcal{S'} \subseteq \mathcal{J}_\cup}
c_{\mathcal{S'}}
\sum_{x = 0}^{|\mathcal{J}_\cap\setminus\mathcal{S'}|}
\sum_{y = 0}^{|\mathcal{S'}\cap\mathcal{J}_\cap|}
 (-)^{y + |\mathcal{S'}\setminus\mathcal{J}_\cap|}2^y
\nonumber\\&&\qquad\qquad\qquad\quad
 \binom{|\mathcal{J}_\cap\setminus\mathcal{S'}|}{x}
 \binom{|\mathcal{S'}\cap\mathcal{J}_\cap|}{y}
\frac{
 \delta^0_{x} - (q+p)^{|\mathcal{S'}\cap\mathcal{J}_\cap| - y} (-q)^{x} 
 }{p^{|\mathcal{S'}| + x + |\mathcal{S'}\cap\mathcal{J}_\cap| - y}}\\
 &=&
\left(\frac{-p}{q-p}\right)^{|\mathcal{J}|+|\mathcal{J}'|}\!\!
\sum_{\mathcal{S'} \subseteq \mathcal{J}_\cup}
c_{\mathcal{S'}}
\sum_{x = 0}^{|\mathcal{J}_\cap\setminus\mathcal{S'}|}
 (-)^{|\mathcal{S'}\setminus\mathcal{J}_\cap|}
\binom{|\mathcal{J}_\cap\setminus\mathcal{S'}|}{x}
 \frac{
 \delta^0_{x}(1-2p)^{|\mathcal{S'}\cap\mathcal{J}_\cap|} - (q-p)^{|\mathcal{S'}\cap\mathcal{J}_\cap|} (-q)^{x} 
 }{p^{|\mathcal{S'}| + x + |\mathcal{S'}\cap\mathcal{J}_\cap|}}
 \nonumber\\&&\qquad\qquad\qquad\quad\\
 &=&
\left(\frac{-p}{q-p}\right)^{|\mathcal{J}|+|\mathcal{J}'|}\!\!
\sum_{\mathcal{S'} \subseteq \mathcal{J}_\cup}
c_{\mathcal{S'}}
(-)^{|\mathcal{S'}\setminus\mathcal{J}_\cap|}
 \frac{
 (1-2p)^{|\mathcal{S'}\cap\mathcal{J}_\cap|} - (-p)^{-|\mathcal{J}_\cap\setminus\mathcal{S'}|}(q-p)^{|\mathcal{J}_\cap\cap\mathcal{S'}|+|\mathcal{J}_\cap\setminus\mathcal{S'}|} 
 }{p^{|\mathcal{S'}| + |\mathcal{J}_\cap\cap\mathcal{S'}|}}\\
 &=&
\left(\frac{1}{q-p}\right)^{|\mathcal{J}|+|\mathcal{J'}|}\!\!
\sum_{\mathcal{S'} \subseteq \mathcal{J}_\cup}
c_{\mathcal{S'}}
(-p)^{|\mathcal{J}_\cup| - |\mathcal{S'}|}
 \left(
 (-p)^{|\mathcal{J}_\cap\setminus\mathcal{S'}|}(1-2p)^{|\mathcal{J}_\cap\cap\mathcal{S'}|} - (q-p)^{|\mathcal{J_\cap}|}
 \right)
\eeqa
using that
\beq
|\mathcal{B}| + |\mathcal{B}'| = |\mathcal{B}_\cup| + |\mathcal{B}_\cap|= 2|\mathcal{B}_\cup| - |\mathcal{B}_\Delta|
\eeq
that $\mathcal{B}_\cup \setminus \mathcal{J}_\cap \subseteq \mathcal{B}_\Delta$, and that $\mathcal{J}_\cup\setminus\mathcal{S'} = \mathcal{J}_\cap\setminus\mathcal{S'}$.
\end{proof}


\bibliography{XTable_mathModel}
\bibliographystyle{unsrt}

\end{document}